\documentclass[journal,onecolumn,11pt]{IEEEtran}

\usepackage{graphicx}
\usepackage{epstopdf}
\usepackage{epsfig}
\usepackage{amsmath,bm,algorithm,algorithmic,float}
\usepackage{amssymb}
\usepackage[thmmarks,amsmath,amsthm,hyperref]{ntheorem}
\usepackage{multirow}
\usepackage{slashbox, color}
\newtheorem{theorem}{Theorem}
\newtheorem{lemma}{Lemma}

\newtheorem{definition}{Definition}
\renewcommand{\baselinestretch}{1.45}
\ifCLASSINFOpdf
\else
\fi
\begin{document}
\title{ Sparse Fast Fourier Transform for Exactly and Generally $K$-Sparse Signals by Downsampling and Sparse Recovery}
\author{Sung-Hsien~Hsieh,
        Chun-Shien~Lu,~\IEEEmembership{Member,~IEEE}
        and~Soo-Chang~Pei,~\IEEEmembership{Fellow,~IEEE} % <-this % stops a space
\thanks{S.-H. Hsieh is with the Institute of Information Science, Academia Sinica,
Taipei 115, Taiwan, and also with the Graduate Institute of Communication Engineering, National Taiwan University, Taipei 106, Taiwan.}
\thanks{C.-S. Lu is with Institute of Information Science, Academia Sinica, Taipei,
Taiwan (e-mail: lcs@iis.sinica.edu.tw).}
\thanks{S.-C. Pei is with the Graduate Institute of Communication Engineering,
National Taiwan University, Taipei 106, Taiwan.}
% <-this % stops a space
% <-this % stops a space
%\thanks{}
}


\markboth{}%
{}



\maketitle

\begin{abstract}
% Fast Fourier transform (FFT) obtains widespread applications in many areas for computing DFT and only costs $O(N \log N)$ calculations.
%How to design an algorithm that outperforms Fast Fourier Transfor (FFT) is a significant challenge.
Fast Fourier Transform (FFT) is one of the most important tools in digital signal processing.
FFT costs $O(N \log N)$ for transforming a signal of length $N$.
%In some applications, FFT of a given input signal is possibly exact $K$-sparse  or general $K$-sparse (with $K$ significant frequencies).
Recently, Sparse Fourier Transform (SFT) has emerged as a critical issue addressing how to compute a compressed Fourier transform of a signal with complexity being related to the sparsity of its spectrum.

In this paper, a new SFT algorithm is proposed for both exactly $K$-sparse signals (with $K$ non-zero frequencies) and generally $K$-sparse signals (with $K$ significant frequencies), with the assumption that the distribution of the non-zero frequencies is uniform.
The nuclear idea is to downsample the input signal at the beginning; then, subsequent processing operates under downsampled signals, where signal lengths are proportional to $O(K)$.
Downsampling, however, possibly leads to ``aliasing''.
By the shift property of DFT, we recast the aliasing problem as complex Bose-Chaudhuri-Hocquenghem (BCH) codes solved by syndrome decoding.
%To further save the computation cost of solving MPP, a top-down iterative strategy where the downsampling factors increase with iterations is proposed for exactly $K$-sparse signals.
The proposed SFT algorithm for exactly $K$-sparse signals recovers $1-\tau$ frequencies with computational complexity $O(K \log K)$ and probability at least $1-O(\frac{c}{\tau})^{\tau K}$ under $K=O(N)$, where $c$ is a user-controlled parameter.

For generally $K$-sparse signals, due to the fact that BCH codes are sensitive to noise, we combine a part of syndrome decoding with a compressive sensing-based solver for obtaining $K$ significant frequencies.
The computational complexity of our algorithm is $\max \left( O(K \log K), O(N) \right)$, where the Big-O constant of $O(N)$ is very small and only a simple operation involves $O(N)$.
Our simulations reveal that $O(N)$ does not dominate the computational cost of sFFT-DT.
% but requires a larger Big-O constant under $K=\Theta(N)$.
% compared to the case of exactly $K$-sparse signals.

In this paper, we provide mathematical analyses for recovery performance and computational complexity, and conduct comparisons with known SFT algorithms in both aspects of theoretical derivations and simulation results.
In particular, our algorithms for both exactly and generally $K$-sparse signals are easy to implement.
%In addition, we also show the upper bound of approximation error in theory.
%In simulations, we show the computational cost outperforming other state-of-the-art algorithms.
\end{abstract}

\begin{IEEEkeywords}
Compressed Sensing, Downsampling, FFT, Sparse FFT, Sparsity.
\end{IEEEkeywords}

% \fi
\IEEEpeerreviewmaketitle

\section{Introduction}\label{sec:intro}
\subsection{Background and Related Work}
%Discrete Fourier transform (DFT) is a fundamental tool for signal processing and costs $O(N^{2})$ calculations, where $N$ is the length of a signal. However, DFT costs too high for real applications when $N$ is large enough.
\IEEEPARstart{F}{ast} Fourier transform (FFT) is one of the most important approaches for fast computing discrete Fourier transform (DFT) of a signal with time complexity $O(N \log N)$, where $N$ is the signal length.
FFT has been used widely in the communities of signal processing and communications.
How to outperform FFT, however, remains a challenge and persistently receives attention.

Sparsity is inherent in signals and has been exploited to speed up FFT in the literature.
A signal of length $N$ is called exactly $K$-sparse if there are $K$ non-zero frequencies with $K < N$.
On the other hand, a signal is called generally $K$-sparse if all frequencies are non-zero but we are only interested in keeping the first $K$-largest (significant) frequencies in terms of magnitudes and ignore the remainder.
Instead of computing all frequencies, Sparse Fourier Transform (SFT) has emerged as a critical topic and aim to compute a compressed DFT, where the time complexity is proportional to $K$.
%the $K$-sparsity of the spectrum of the signal.

A. C. Gilbert \cite{Gilbert2014} {\em et al.} propose an overview of SFT and summarize a common three-stage approach:
1) identify locations of non-zero or significant frequencies;
2) estimate the values of the identified frequencies; and
3) subtract the contribution of the partial Fourier representation computed from the first two stages from the signal and go back to stage 1.
Some prior works are briefly described as follows.

M. A. Iwen \cite{Iwen2010} proposes a sublinear-time SFT algorithm based on Chinese Remainder Theorem (CRT) with computational complexities
(a) $O(K \log^{5} N)$ with a non-uniform failure probability per signal and
(b) $O(K^2 \log^4 N)$ with a deterministic recovery guarantee.
Iwen's algorithm can work for general $N$ with the help of interpolation.
Although the algorithm offers strong theoretical analysis, the empirical experiments show that it suffers Big-O constants. For example, in Fig. 5 of \cite{Iwen2010}, it shows to outperform FFTW under $K=8$ and $N=2^{18}$.
The approximation error bounds in \cite{Iwen2010} are further improved in \cite{Iwen2013}.

H. Hassanieh {\em et al.} propose so-called Sparse Fast Fourier Transform (sFFT) \cite{Haitham2012}\cite{Haitham2012_1}. %which was proved to outperform FFT.
The idea behind sFFT is to subsample fewer frequencies (proportional to $K$) since most of frequencies are zero or insignificant.
Nevertheless, the difficulty is which frequencies should be subsampled as the locations and values of the $K$ non-zero frequencies are unknown. To cope with this difficulty, sFFT utilizes the strategies of filtering and permutation introduced in \cite{Gilbert2005}, which can increase the probability of capturing useful information from subsampled frequencies.
For exactly $K$-sparse and general $K$-sparse signals, sFFT costs $O (K \log N)$ and $O (K \log N \log \frac{N}{K})$, respectively.
In their simulations, sFFT is faster than FFTW \cite{Frigo2005} (a very fast C subroutine library for computing FFT) for exactly $K$-sparse signals with $K\leq\frac{N}{2^{6}}$.

%Nevertheless, for a generally $K$-sparse signal, subsampled frequencies certainly face interference from neighboring insignificant frequencies.
%To alleviate this interference, which is assumed to be random noise in \cite{Haitham2012}, sFFT permutes $\bm{X}$ randomly in order to make neighboring entries randomly distributed.
%In other words, interference is considered to be random noise.
%In this way, sFFT produces a number of subsampled signals under different permutations, where each subsampled signal is considered a candidate and the median is identified from these candidates.

Even though sFFT \cite{Haitham2012}\cite{Haitham2012_1} is outstanding, there are some limitations, summarized as follows:
1) Filtering and permutation are operated on the input signal. These operations are related to $N$.
Thus, the complexity of sFFT still involves $N$ and cannot achieve the theoretical ideal complexity $O(K \log K)$.
2) sFFT only guarantees that it succeeds with a constant probability ({\em e.g.}, $2/3$).
3) The implementation of sFFT for generally $K$-sparse signals is very complicated as it involves too many parameters that are difficult to set.\footnote{In fact, according to our private communication with the authors of \cite{Haitham2012}\cite{Haitham2012_1}, they would not recommend implementing this code since it is not trivial. The authors also suggest that it is not easy to clearly illustrate which setting will work best because of the constants in the Big-O functions and because of the dependency on the implementation. The authors themselves did not implement it since they believed that the constants would be large and that it would not realize much improvement over FFTW.}

Ghazi {\em et al.} \cite{Ghazi2013} propose another algorithm based on Prony's method for exactly $K$-sparse signals.
The basic idea is similar to our previous work \cite{Hsieh2013}.
%Both methods first downsample original signals before recovering $K$ non-zero frequencies from the downsampled signals via error correction techniques, where \cite{Ghazi2013} uses syndrome decoding, which is equivalent to the Moment-preserving problem considered in \cite{Hsieh2013}.
%Ghazi's method first downsamples original signals and then recovers $K$ non-zero frequencies from downsampled signals based on error correction techniques of Reed-Solomon code, which is equivalent to moment-preserving thresholding proposed by \cite{Hsieh2013}.
The key difference is that Ghazi {\em et al.}'s method recovers all $K$ non-zero frequencies once, while we propose a top-down strategy to solve $K$ non-zero frequencies iteratively. Furthermore, due to different parameter settings and root finding algorithms, Ghazi's SFT costs $O(K \log K+K(\log\log N)^{O(1)})$ along with different big-O constants. The comparison between these two methods in terms of computational complexity and recovery performance will be discussed later in Sec. \ref{ssec:howtochoose_a_and_d}.

S. Heider {\em et al.}'s method \cite{Heider2013} combines Prony-like methods with quasi random sampling and band pass filtering.
Compared with our method, they estimate the positions and values of non-zero frequencies in each band based on the ESPRIT method instead of syndrome decoding.
ESPRIT requires more computational cost resulting in the total complexity being $O(K^{\frac{5}{3}} \log^{2} N)$.
Their proof also shows $K=O(N^{0.5})$ that is more strict than $K=O(N)$ in our case for exactly-$K$ sparse signal.

Pawar and Ramchandran \cite{Pawar2013} propose an algorithm, called FFAST (Fast Fourier Aliasing-based Sparse Transform), which focuses on exactly $K$-sparse signals.
Their approach is based on filterless downsampling of the input signal using a constant number of co-prime downsampling factors guided by CRT.
These aliasing patterns of different downsampled signals are formulated as parity-check constraints of good erasure-correcting sparse-graph codes.
FFAST costs $O(K \log K)$ but relies on the constraint that co-prime downsampling factors must divide $N$.
Moreover, the smallest downsampling factor bounds FFAST's computational cost.
For example, if $N=2^{20}3^{2}$ and $K=2^{16}$, the smallest downsampling factor is $3^{2}$.
In this case, the computational cost of calculating FFT of a downsampled signal with length $\frac{N}{3}^{2}$ is higher than $O(K \log K)$.
Actually, these limitations are possibly harsh.

We summarize and compare the SFT algorithms reviewed above in Table \ref{Table: SFT comparision} in terms of the number of samples, computational complexity, and assumption regarding sparsity.
More specifically, the number of samples decides how much information SFT algorithms require in order to reconstruct $K$-sparse signals.
It is especially important for some applications, including Analog-to-Digital converter, which are benefited by low sampling rates.
Moreover, the assumption of a certain range of sparsity guarantees that SFT algorithms can have high quality of reconstruction.
We can find from Table \ref{Table: SFT comparision} that our algorithms have the lowest computational complexity, the lowest number of samples, and the best range of sparsity for exactly $K$-sparse signals.
Although the sparsity constraint $K=\Theta(N)$ seems to be more tough for generally $K$-sparse signals in our method,
%It seems that the proposed algorithm cannot work well when the signal is very sparse.
for a (very) sparse signal we still can solve it by assuming that its sparsity is higher than the true one with more computational cost.
In the simulations, we show that the Big-O constants for both exactly $K$-sparse and generally $K$-sparse signals are actually small, implying the practicability of our proposed approaches for real implementation.

\begin{table}[t]
\fontsize{7.5pt}{1em}\selectfont
\centering
\setlength{\abovecaptionskip}{0pt}
\setlength{\belowcaptionskip}{4pt}
\caption{Comparison between SFT algorithms in terms of computational complexity, required samples and assumptions.}
\label{Table: SFT comparision}
\doublerulesep=0pt
\begin{tabular}[tc]{|c||c|c|c|c|c|c|}
\hline
 \multirow{2}{*}{}&    \multicolumn{3}{c|}{Exactly $K$-sparse signal}&  \multicolumn{3}{c|}{Generally $K$-sparse signal}  \\
\cline{2-7} & Samples  & Complexity & Assumption & Samples & Complexity & Assumption   \\ \hline\hline
\cite{Iwen2010}& $O(K \log^{4} N)$  &$O(K \log^{5} N)$ & $K=O(N)$ & $O(K \log^{4} N)$& $O(K \log^{5} N)$&  $K=O(N)$  \\ \hline
\cite{Haitham2012_1}& $O(K )$ & $O(K \log N)$ & $K=O(N)$ & $O(K \log(\frac{N}{K})/\log\log N )$ & $O (K \log N \log \frac{N}{K})$ &  $K=O(N)$  \\ \hline
\cite{Ghazi2013}& $O(K)$ & $O(K \log K+K(\log\log N)^{O(1)})$&  $K=O(N)$ &$O(K \log N)$ & $O(K \log^{2} N)$ & $K=\Theta(\sqrt{N})$    \\ \hline
\cite{Heider2013}& $O(K)$ & $O(K^{\frac{5}{3}} \log^{2} N)$ & $K=O(\sqrt{N})$ & void & void & void    \\ \hline
\cite{Pawar2013}& $O(K)$ & $O(K \log K)$& $K=O(N^{\alpha})$, $\alpha<1$ & void & void & void   \\ \hline
This paper& $O(K)$ & $O(K \log K)$& $K=O(N)$ & $O(K)$  & $O(K \log K)$ & $K=\Theta(N)$   \\ \hline
\end{tabular}
\end{table}

\subsection{Our Contributions}
In our previous work \cite{Hsieh2013}, we propose a SFT algorithm, called sFFT-DT, based on filterless downsampling with time complexity of $O(K \log K)$ only for exactly $K$-sparse signals.
The idea behind sFFT-DT is to downsample the input signal in the time domain before directly conducting all subsequent operations on the downsampled signals.
By choosing an appropriate downsampling factor to make the length of a downsampled signal be $O(K)$, no operations related to $N$ are required in sFFT-DT.
Downsampling, however, possibly leads to ``aliasing,'' where different frequencies become indistinguishable in terms of their locations and values.
To overcome this problem, the locations and values of these $K$ non-zero entries are considered as unknown variables and the ``aliasing problem'' is reformulated as ``Moment Preserving Problem (MPP)''.
Furthermore, sFFT-DT is conducted in a manner of a top-down iterative strategy under different downsampling factors, which can efficiently reduce the computational cost.
In comparison with other CRT-based approaches \cite{Heider2013}\cite{Pawar2013} that require multiple co-prime integers dividing $N$, our method only needs the downsampling factor to divide $N$ but does not suffer the co-prime constraint, implying that sFFT-DT has more freedom for $N$.

In this paper, we further examine the accurate computational cost and theoretical performance of sFFT-DT for exactly $K$-sparse signals.
%The difference between sFFT-DT and Ghazi {\em et al.}'s sFFT \cite{Ghazi2013} is pointed out since both methods are based on the same idea.
We derive the Big-O constants of computational complexity of sFFT-DT and show that they are smaller than those of Ghazi {\em et al.}'s sFFT \cite{Ghazi2013}.
In addition, sFFT-DT is efficient due to $K=O(N)$, which makes it useful whatever the sparsity $K$ is.
Finally, all operations of sFFT-DT are solved via analytical solutions but those of Ghazi {\em et al.}'s sFFT involve a numerical root finding algorithm, which is more complicated in terms of hardware implementation.

In the context of SFT, sparsity $K$ plays an important role.
The performance and computational complexity of previous SFT algorithms \cite{Haitham2012}\cite{Haitham2012_1}\cite{Ghazi2013}\cite{Pawar2013} have been analyzed based on the assumption that sparsity $K$ is known in advance.
In practice, however, $K$ is unknown and is an input parameter decided by the user.
If $K$ is not guessed correctly, the performance is degraded and/or the computational overhead is higher than expected because the choice of some parameters depends on $K$.
In this paper, we propose a simple solution to address this problem and relax this impractical assumption.
We show that the cost for deciding $K$ is the same as that required for sFFT-DT with known $K$.

In addition to conducting more advanced theoretical analyses, we also study sFFT-DT for generally $K$-sparse signals in this paper.
For generally $K$-sparse signals, since all frequencies are non-zero, each frequency of a downsampled signal is composed of significant and insignificant frequencies due to aliasing.
To extract significant components from each frequency, the concept of sparse signal recovered from fewer samples, originating from compressive sensing (CS) \cite{Donoho2006}, is employed since significant entries are ``sparse''. A pruning strategy is further used to exclude locations of insignificant terms.
We prove the sufficient conditions of robust recovery, which means reconstruction error is bounded, with time complexity $\max(O(K\log K), O(N))$ under $K=\Theta(N)$.
The empirical experiments show that the Big-O constant of sFFT-DT is small and outperforms FFT when $N=2^{24}$ and $K\leq 2^{16}$.

Finally, we conclude that our methods are easy to implement and are demonstrated to outperform the state-of-the-art in terms of theoretical analyses and simulation results.

\subsection{Organization of This Paper}
The remainder of this paper is organized as follows.
%First, notations frequently used in the paper are summarized in Sec. \ref{Sec: Preliminary}.
In Sec. \ref{Sec: sFFT-DT: Exact K-Sparse}, we describe the proposed method for exactly $K$-sparse signals.
Our method for generally $K$-sparse signals will be expounded in Sec. \ref{Sec: sFFT-DT: General K-Sparse}.
Conclusions are provided in Sec. \ref{Sec: Conclusions}.

%\section{Preliminary}\label{Sec: Preliminary}
%
%Table \ref{Table: Notation Table} summarizes the notations frequently used in this paper.
%Bold font is used for representing matrix or vector.
%First, we briefly introduce some DFT properties and notations used in this paper.

\section{sFFT-DT for Exactly $K$-Sparse Signals}\label{Sec: sFFT-DT: Exact K-Sparse}
We describe the proposed method for exactly $K$-sparse signals and provide analyses for parameter setting, computational complexity, and recovery performance.
The proposed method contains three steps.
\begin{itemize}
  \item[1.] Downsample the original signal in the time domain.
  \item[2.] Calculate Discrete Fourier Transform (DFT) of the downsampled signal by FFT.
  \item[3.] Use the DFT of the downsampled signal to locate and estimate $K$ non-zero frequencies..
\end{itemize}
Steps 1 and 2 are simple and straightforward.
Thus, we focus on Step 3 here.

Throughout the paper, common notations are defined as follows.
Let $\bm{x}\in \mathbb{C}^{N}$ be the input signal in the time domain, and let $\hat{\bm{x}} \in \mathbb{C}^{N} $ be DFT of $\bm{x}$.
$F\in \mathbb{C}^{N\times N}$ is the DFT matrix such that $\hat{\bm{x}}=F\bm{x}$ with $F_{k,l}=e^{\frac{-i2\pi kl}{N}}/N$ and $F_{k,l}^{-1}=e^{\frac{i2\pi kl}{N}}$.

\subsection{Problem Formulation}\label{ssec:DFT Property}
Let $\bm{x_{d}}$ be the signal downsampled from an original signal $\bm{x}$, where $x_{d}[k]=x[dk]$, $k\in [0,\frac{N}{d}-1]$, and integer $d\geq 1$  is a downsampling factor.
The length of the downsampled signal $\bm{x_{d}}$ is $\frac{N}{d}$
Let $\bm{\hat{x}_{d}}$ be DFT of $\bm{x_{d}}$, where
\small
\begin{equation}
\begin{aligned}
	  \hat{x}_{d}[k]=&(\hat{x}[k]+\hat{x}[k+\frac{N}{d}]+\hat{x}[k+2\frac{N}{d}]+...+\hat{x}[k+(d-1)\frac{N}{d}])/d.
\end{aligned}
\label{eq:xd}
\end{equation}
\normalsize
The objective here is to locate and estimate $K$ non-zero frequencies of $\bm{\hat{x}}$ from $\bm{\hat{x}_{d}}$.

Note that each frequency of $\bm{\hat{x}_{d}}$ is a sum of $d$ terms of $\bm{\hat{x}}$.
When more than two terms of $\bm{\hat{x}}$ are non-zero, ``aliasing'' occurs, as illustrated in Fig. \ref{fig:Iterative Pyramid}.
% illustrates a simple case when aliasing occurs.
Fig. \ref{fig:Iterative Pyramid}(a) shows an original signal in the frequency domain, where only three frequencies are non-zero (appearing at normalized frequencies = $0\pi$, $0.5\pi$, and $\pi$).
Fig. \ref{fig:Iterative Pyramid}(b) shows the downsampled signal in the frequency domain when $d=2$, where the downsampled frequency at $0\pi$ incurs aliasing; {\em i.e.}, the frequency of $\bm{\hat{x}}$ at $0\pi$ collides with the one at $\pi$.
In Fig. \ref{fig:Iterative Pyramid}(b), we solve all non-zero downsampled frequencies once, no matter whether aliasing occurs or not.
This procedure is called non-iterative sFFT-DT and will be discussed in detail later.
Instead of solving all of the downsampled frequencies once, Fig. \ref{fig:Iterative Pyramid}(c) illustrates an example of iteratively solving frequencies.
At the first iteration, the downsampled frequency without aliasing at $1\pi$ is solved.
This makes the remaining downsampled frequencies more sparse.
Then, the signal is downsampled again with $d=4$.
At the second iteration, we solve the downsampled frequency with aliasing at $0\pi$.
This procedure, called iterative sFFT-DT, will be discussed further in Sec. \ref{ssec:ItMethod}.
%To achieve the goal mentioned above, we must solve the aliasing problem of each frequency of $\bm{\hat{x}_{d}}$.

%, benefits the computation cost compared with non-iterative sFFT-DT in Fig. \ref{fig:Iterative Pyramid}(b). Detailed analysis in terms of computation complexity and performance will be shown in Sec. \ref{ssec:howtochoose_a_and_d} and \ref{ssec:conv_analysis}.

\begin{figure}[h]
  \centering{\epsfig{figure=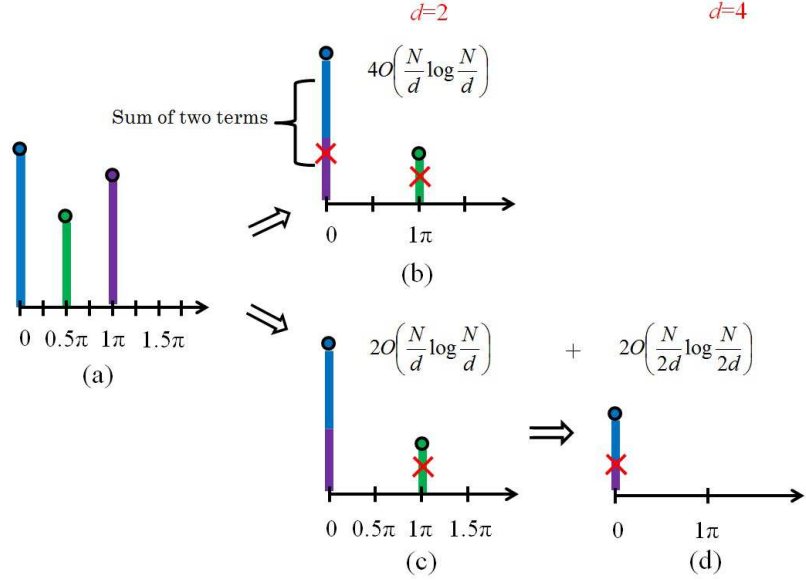,width=4.25 in}}
\hfill
\caption{Aliasing and its iterative solver. (a) Original signal in frequency domain. (b) Downsampled signal in frequency domain with $d=2$. If we want to solve all frequencies once, it requires $4$ FFTs. (c) Similar to (b), however, the frequency (at normalized frequency $1\pi$) at $d=2$ is solved first and requires $2$ FFTs. (d) Remaining frequency (at $0\pi$) requires $2$ extra FFTs at $d=4$.}
\label{fig:Iterative Pyramid}
\end{figure}

In the following, we describe how to solve the aliasing problem by introducing the shift property of DFT.
Let $x_{d,l}[k]=x[dk+l]$, where $l$ denotes the shift factor.
Each frequency of $\bm{\hat{x}_{d,l}}$ is denoted as:
\small
\begin{equation}
\begin{aligned}
     \hat{x}_{d,l}[k]&=(\hat{x}[k]F_{k,l}^{-1}+\hat{x}[k+\frac{N}{d}]F_{k+\frac{N}{d},l}^{-1}+...+\hat{x}[k+(d-1)\frac{N}{d}]F_{k+(d-1)\frac{N}{d},l}^{-1})/d.
\end{aligned}
\label{eq:xd with shift}
\end{equation}
\normalsize
Thus, Eq. (\ref{eq:xd with shift}) degenerates to Eq. (\ref{eq:xd}) when $l=0$. In practice, all we can obtain are $ \hat{x}_{d,l}[k]$'s for different $l$'s.

For each downsampling factor $d$, there will be no more than $d$ terms on the right side of Eq. (\ref{eq:xd with shift}), where
%How many terms are on the right side of Eq. (\ref{eq:xd with shift}) are unknown.
each term contains two unknown variables, $\hat{x}[k]$ and $F_{k,l}^{-1}$.
Let $a$, $1\leq a\leq d$, denote the number of terms on the right side of Eq. (\ref{eq:xd with shift}).
Therefore, we need $2a$ equations to solve these $2a$ variables, and $l$ is within the range of $[0,2a-1]$.
%Let $a\in [1 \d ]$ be the number of terms.
%There are a total of $2a$ unknown variable. In other words, at least $2a$ equations are required to solve $2a$ unknown variables.
By taking the above into consideration, the problem of solving the $2a$ unknown variables on the right side of Eq. (\ref{eq:xd with shift}) can be formulated\footnote{In the previous version \cite{Hsieh2013}, it is interpreted as a moment preserving problem (MPP). Specifically, solving MPP is equivalent to solving complex BCH codes, where the syndromes produced by partial Fourier transform are consistent with moments.} via BCH codes as:
%we have the following system of equations:
\small
\begin{equation}
\begin{aligned}
    m_{0}&=p_{0}z_{0}^{0}+p_{1}z_{1}^{0}+...+p_{a-1}z_{a-1}^{0}, \\
    m_{1}&=p_{0}z_{0}^{1}+p_{1}z_{1}^{1}+...+p_{a-1}z_{a-1}^{1}, \\
                                                   &\vdots   \\
    m_{2a-1}&=p_{0}z_{0}^{2a-1}+p_{1}z_{1}^{2a-1}+...+p_{a-1}z_{a-1}^{2a-1},
\end{aligned}
\label{eq:MPT obejective fun}
\end{equation}
\normalsize
where $\hat{x}_{d,l}[k]$ is known and denoted as $m_{l}$  while $p_{j}$ and $z_{j}^{l}$ represent unknown $\hat{x}[s_{j}]$ and $F_{s_{j},l}^{-1}$, respectively, for $s_{j} \in \{k, \  k+\frac{N}{d} \  , ...\  ,k+(d-1)\frac{N}{d}\}$ and $j\in [0,a-1]$. To simplify the notation, we let $S_{k}= \{k, \  k+\frac{N}{d} \  , ...\  ,k+(d-1)\frac{N}{d}\}$ and $U_{k}=\{ F_{k,l}^{-1} , \  F_{k+\frac{N}{d},l}^{-1} \  , ...\  ,  F_{k+(d-1)\frac{N}{d},l}^{-1} \}$.
%Notice that only $m_{i}$'s are known.

It is trivial that no aliasing occurs if $a=1$, irrespective of the downsampling factor.
Under this circumstance, we have $m_{0}=\hat{x}_{d,0}[k]$, $m_{1}=\hat{x}_{d,1}[k]$, $m_{0}=p_{0}z_{0}^{0}=\hat{x}[s_{0}]/d$, and $m_{1}=p_{0}z_{0}^{1}=\hat{x}[s_{0}]e^{i2\pi s_{0}/N}/d$, according to Eq. (\ref{eq:MPT obejective fun}).
We obtain that $|m_{0}|=|\hat{x}[s_{0}]|/d=|m_{1}|$ and $m_{1}/m_{0}=e^{i2\pi s_{0}/N}$.
After some derivations, we can solve $s_{0}$ and assign $\hat{x}[s_{0}]=d\hat{x}_{d,0}[k]$ at the position $s_0$.
The above solver only works under a non-aliasing environment with $a=1$.
Nevertheless, when aliasing appears ({\em i.e.}, $a>1$), it fails.

To solve the aliasing problem, it is observed from Eq. (\ref{eq:MPT obejective fun}) that all we know are $ m_{i}$'s for $0\leq i\leq 2a-1$, called syndromes in BCH codes.
Thus, we utilize syndrome decoding \cite{MacWilliams11977}, which is also equivalent to the solver presented in Ghazi {\em et al.}'s sFFT.
Syndrome decoding is discussed in the next subsection.

\subsection{Syndrome Decoding}\label{ssec:MPT}
%Moment-preserving thresholding (MPT) was originally developed for image thresholding in \cite{Tsai1985}.
Note that Eq. (\ref{eq:MPT obejective fun}) is nonlinear and cannot be solved by simple linear matrix operations.
On the contrary, we have to solve $z_j$'s first, such that Eq. (\ref{eq:MPT obejective fun}) becomes linear.
Then, $p_{i}$'s can be solved by matrix inversion.
Thus, the main difficulty is how to solve $z_j$'s given known syndromes.
According to \cite{Szego1975}, given the unique syndromes with $m_{0}$, $m_{1}$, ..., $m_{2a-1}$, there must exist the corresponding orthogonal polynomial equation, $P(z)$, with roots $z_{j}$'s for $0\leq j\leq a-1$.
%In addition, the relationship between the moments and $P(z)$ is known.
%Thus, $P(z)$ can be obtained given the moments.
That is, $z_{j}$'s can be obtained as the roots of $P(z)$.
The steps for syndrome decoding are as follows. \\
Step (i): Let the orthogonal polynomial equation $P(z)$ be:
\small
\begin{equation}
P(z)=z^{a}+c_{a-1}z^{a-1}+...+c_{1}z+c_{0}.
\label{eq:poly fun of MPT}
\end{equation}
\normalsize
The relationship between $P(z)$ and the syndromes is as follows:
\small
\begin{equation}
\begin{aligned}
    -m_{a}&=c_{0}m_{0}+c_{1}m_{1}+...+c_{a-1}m_{a-1}, \\
    -m_{a+1}&=c_{0}m_{1}+c_{1}m_{2}+...+c_{a-1}m_{a}, \\
                                                   &\vdots   \\
    -m_{2a-1}&=c_{0}m_{a-1}+c_{1}m_{a}+...+c_{a-1}m_{2a-2}.
\end{aligned}
\label{eq:auxiliary fun of MPT}
\end{equation}
\normalsize
Eq. (\ref{eq:auxiliary fun of MPT}) can be formulated as $\bm{m}=\bm{M}\bm{c}$, where $\bm{M}_{i,j}= m_{i+j}$, $\bm{c}=[c_{0} \ c_{1} \ ... \ c_{a-1} ]^T$, and $\bm{m}=[-m_{a} \ -m_{a+1} \ ... \ -m_{2a-1} ]^T$.
Thus, Eq. (\ref{eq:auxiliary fun of MPT}) can be solved by matrix inversion $\bm{M}^{-1}$ to obtain $c_{j}$'s.\\
Step (ii): Find the roots of $P(z)$ in Eq. (\ref{eq:poly fun of MPT}).
%\begin{equation}
%   z^{a}+c_{a-1}z^{a-1}+...+c_{1}z+c_{0}=0.
%\label{eq:poly fun of MPT}
%\end{equation}
These roots are the solutions of $z_{0}$, $z_{1}$,...$z_{a-1}$, respectively.\\
Step (iii): Substitute all $z_{j}$'s into Eq. (\ref{eq:MPT obejective fun}), and solve the resulting equations to obtain $p_{j}$'s.

Tsai \cite{Tsai1985} showed a complete analytical solution composed of the aforementioned three steps for $a\leq 4$, based on the constraint that $p_{0}+p_{1}+...+p_{a-1}=1$.
%In our case, solving MPP is equivalent to solving complex BCH codes where syndromes produced by partial Fourier transform are consistent to moments.
%Thus, we use the solution proposed by \cite{Tsai1985}.
Nevertheless, for the aliasing problem considered here, the constraint is $p_{0}+p_{1}+...+p_{a-1}=\hat{x}_{d,0}[k]$, as indicated in Eq. (\ref{eq:xd with shift}).
We have also derived the complete analytical solution accordingly for $2\leq a\leq 4$.
Please see Appendix in Sec. \ref{Sec: Appendix}.
%\footnote{\textcolor{red}{The solutions for $a=2$ to $a=4$ are described in the Appendix.}}
%In fact, $c_{i}$'s are calculated by matrix inversion since Eq. (\ref{eq:auxiliary fun of MPT}) denotes a system of linear equations.
%A close-form solution by factorization of Eq. (\ref{eq:poly fun of MPT}) exists for $z_{i}$'s when $a \leq 4$.
%After $z_{i}$'s are solved, Eq. (\ref{eq:MPT obejective fun}) becomes a linear equation.
%Thus, $p_{i}$'s can be solved like $c_{i}$'s.
The analytical solutions for a univariate polynomial with $a\leq 4$ cost $O(a^2)$ operations.
%It also holds for  univariate polynomial with order being less than or equal to 4 ($a \leq 4$).
Since there are $\frac{N}{d}$ frequencies, the computational cost of syndrome decoding is $O(\frac{N}{d}a^2)$.
For $a > 4$, Step (i) still costs $O(a^2)$, according to the Berlekamp-Massey algorithm \cite{Massey1963}, which is well-known in Reed-Solomon decoding \cite{MacWilliams11977}.
In addition, Step (iii) is designed to calculate the inverse matrix of a Vandermonde matrix and costs $O (a^2)$ \cite{NChen2008}.
There is, however, no analytical solution of Step (ii) for $a > 4$.
Thus, numerical methods of root finding algorithms with finite precision are required.
%It requires more computation operation than $O(a^2)$.
A fast algorithm proposed by Pan \cite{Pan2002} can approximate all of the roots with $O(a(\log\log N)^{O(1)})$, where the detailed proof was shown in \cite{Ghazi2013}.
If $(\log\log N)^{O(1)}>a$, Step (ii) will dominate the cost of syndrome decoding.

It is noted that the actual number of collisions for each frequency, $a$ ($1\leq a\leq d$), is unknown in advance.
In practice, we choose a maximum number of collision $a_{m}$ and expect $a\leq a_{m}$ for all downsampled frequencies. Under the circumstance, $2a_{m}$ syndromes are required for syndrome decoding.
If $a$'s of all downsampled frequencies are smaller than or equal to $a_{m}$, the syndrome decoding perfectly recovers all of the frequencies; {\em i.e.,} it resolves all non-zero values and locations of $\bm{\hat{x}}$.
Otherwise, the non-zero entries of $\bm{\hat{x}}$ cannot be recovered due to insufficient information.
Although a larger $a_{m}$ guarantees better recovery performance, it also means that more syndromes and higher computational cost are required.

In sum, the cost of syndrome decoding consists of two parts. Since the size of a downsampled signal is $\frac{N}{d}$, the cost of generating the required syndromes via FFT is $O(2a_{m}\frac{N}{d}\log \frac{N}{d})$, which is called the ``P1 cost of syndrome decoding'' hereafter.
Second, as previously mentioned, solving the aforementioned Steps (i), (ii), and (iii) will cost $ O(\frac{N}{d}a_{m}^2)$ for $a_{m} \leq 4$ and cost $O(\frac{N}{d}a_{m}(\log\log N)^{O(1)})$ for $a_{m} > 4$, where either of which is defined as the ``P2 cost of syndrome decoding''.
Lemma \ref{lemma:cost of non-iterative sFFT-DT} summarizes the computational cost of syndrome decoding.
\begin{lemma}
\label{lemma:cost of non-iterative sFFT-DT}
Give $a_{m}$ and $d=O(\frac{N}{K})$, sFFT-DT, including generating syndromes by FFTs and syndrome decoding, totally costs $O(a_{m}\frac{N}{d}\log \frac{N}{d})$ for $a_{m}\leq 4$ and $O(a_{m}\frac{N}{d}\log \frac{N}{d}+a_{m}\frac{N}{d}(\log \log N)^{O(1)})$ for $a_{m} > 4$.
\end{lemma}

So far, our method of solving all downsampled frequencies is based on fixing downsampling factor $d$ (and $a_{m}$), as an example illustrated in Fig. \ref{fig:Iterative Pyramid} (b).
In this case, we call this approach, non-iterative sFFT-DT.
Its iterative counterpart, iterative sFFT-DT,
%which is the focus of this paper,
will be described later in Sec. \ref{ssec:ItMethod} and Sec. \ref{ssec:exact_algorithm}.

%Following the discussions in the previous subsection, the P2 cost of MPP, in fact, is equivalent to solving the error locator polynomial of Reed-Solomon decoding in Ghazi {\em %et al.}'s sFFT \cite{Ghazi2013}.
%Thus, the computational cost of Ghazi {\em et al.}'s sFFT, similar to ours, still is composed of the P1 cost and P2 cost of MPP.
%The main difference between our non-iterative sFFT-DT and Ghazi {\em et al.}'s sFFT \cite{Ghazi2013} is how to set $a_{m}$ and $d$.
%In non-iterative sFFT-DT, $a_{m}=4$ is fixed and an appropriate $d$ is selected to make $a$'s of most downsampled frequencies smaller than or equal to $a_{m}$.
%Under this situation, the total cost is $O(2a_{m}\frac{N}{d}\log \frac{N}{d}+\frac{N}{d}a_{m}^2)$.
%On the other hand, Ghazi {\em et al.}'s sFFT is to set $d=\frac{N\log K }{K}$ and $a_{m} = C\log K$ ($C$ is a constant), and the total cost is $O(2a_{m}\frac{N}{d}\log %\frac{N}{d}+\frac{N}{d}a_{m}(\log\log N)^{O(1)})$.
%In fact, these two kinds of parameter settings lead to different computational cost and recovery performance.
%We will discuss this issue further in the next subsection.

%Moreover, in order to reduce computational cost in our framework further, a top-down iterative strategy for sFFT-DT will be presented in Sec. \ref{ssec:ItMethod}.

\subsection{Analysis}\label{ssec:howtochoose_a_and_d}

%As described in Sec. \ref{ssec:MPT}, the cost in non-iterative sFFT-DT is $O(2a_{m}\frac{N}{d}\log\frac{N}{d}+\frac{N}{d}a_{m}^2)$ and one in \cite{Ghazi2013} is $O(2a_{m}\frac{N}{d}\log \frac{N}{d}+\frac{N}{d}a_{m}(\log\log N)^{O(1)})$.
In this section, we first will study the relationship between $a_{m}$ and $d$, and analyze the probability of a downsampled frequency with number of collisions larger than $a_m$.
Second, we will discuss computational complexity and recovery performance of our non-iterative sFFT-DT.
Third, we will compare non-iterative sFFT-DT with Ghazi {\em et al.}'s sFFT \cite{Ghazi2013}.
In addition, the Big-O constant of complexity is induced in order to highlight the computational simplicity of non-iterative sFFT-DT.
Finally, we will conclude by presenting an iterative sFFT-DT approach to reduce computational cost further.

\subsubsection{Relationship between Maximum Number of Collisions and Downsampling Factor}
Now, we consider the relationship between $a_m$ and $d$.
If $a_{m}$ is set to $d$, then we always can recover any $\bm{\hat{x}}$ without errors but the computational cost
will be larger than that of FFT.
%computational overheads is beyond FFT.
Thus, it is preferable to set smaller $a_{m}$, which is still feasible when $\bm{\hat{x}}$ is uniformly distributed.
For each frequency, the number of collisions, $a$, will be small with higher probability if $\frac{dK}{N}$ is small enough, as Lemma \ref{lemma:probability of bin} illustrates
\begin{lemma}
\label{lemma:probability of bin}
Suppose $K$ non-zero entries distribute uniformly ({\em i.e.}, with probability $\frac{K}{N}$) in $\bm{\hat{x}}$.
Let $Pr(d,a_{m})$ denote the probability that there is at least a downsampled frequency with number of collisions $a > a_{m}$ when the downsampling factor is $d$.
Then, $Pr(d,a_{m})\leq \frac{N}{d} (\frac{deK}{N(a_{m}+1)})^{a_{m}+1}$, where $e$ is Euler's.
And non-iterative sFFT-DT obtains perfect recovery with probability at least $\rho=1-Pr(d,a_{m})$.
\end{lemma}

\begin{proof}
For each downsampled frequency, the probability of $a > a_{m}$ is ${ \sum_{i=a_{m}+1}^{d} \binom{d}{i} (\frac{K}{N})^{i} (1-\frac{K}{N})^{d-i}} $, which is smaller than $\binom{d}{a_{m}+1} (\frac{K}{N})^{a_{m+1}}$.
Under this circumstance, the probability of at least a downsampled frequency with $a > a_{m}$ is  bounded  by $\binom{\frac{N}{d}}{1} \binom{d}{a_{m}+1} (\frac{K}{N})^{a_{m+1}}$.
Thus, we can derive:
\small
\begin{equation}
\begin{aligned}
Pr(d,a_{m}) &\leq \frac{N}{d} \binom{d}{a_{m}+1} (\frac{K}{N})^{a_{m+1}}\leq \frac{N}{d} (\frac{dK}{N})^{a_{m}+1}  \frac{1}{(a_{m}+1)!} \\
&\leq \frac{N}{d} (\frac{dK}{N})^{a_{m}+1}  (\frac{e}{a_{m}+1})^{a_{m}+1} = \frac{N}{d} (\frac{deK}{N(a_{m}+1)})^{a_{m}+1}.
\end{aligned}
\label{Eq: lemma:probability of bin}
\end{equation}
\end{proof}
\normalsize
The probability that $\bm{\hat{x}}$ can be perfectly reconstructed using sFFT-DT is $1-Pr(d,a_{m})$ since $a > a_{m}$ results in the fact that the syndrome decoding cannot attain the correct values and locations in the frequency domain.
Furthermore, since $Pr(d,a_{m})$ is controlled by $\frac{N}{d}$, $K$, and $a_{m}$, it can be very low based on an appropriate setting.
Let $\displaystyle N^{+}=\frac{N}{dK}$ denote the ratio of the length ($\displaystyle \frac{N}{d}$) of a downsampled signal to $K$.
Our empirical observations, shown in Fig. \ref{fig:aliasing}, indicate the probability of collisions at different $N^{+}$'s.
For $a>4$, the probability of collisions is very close to $0$.
\begin{figure}[h]
  \centering{\epsfig{figure=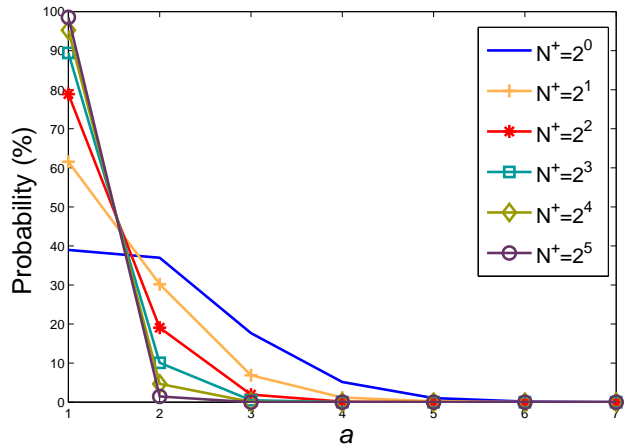,width=3.7in}}
\hfill
\caption{The probability of collisions for $1\leq a\leq 7$ at different $N^{+}$'s, where $a$ denotes the number of collisions. The results show that $a > 2$, in fact, seldom occurs.}
\label{fig:aliasing}
\end{figure}
%Since $\frac{N}{d}$ and $a_{m}$ are user-defined, it can be very low based on an appropriate setting, which is further proven in Theorem \ref{theorem:probability of sfftdt v1}.

%\begin{figure}[!t]
%  \centering{\epsfig{figure=Aliasing.eps,width=2.55in}}
%\hfill
%\caption{The distribution of all $a$'s at different $N^{+}$'s, where $a$ denotes the number of aliasing terms on the right side of Eq. (2). The results show that aliasing ($a > 2$), in fact, seldom occurs.}
%\label{fig:aliasing}
%\end{figure}

\subsubsection{Computational Cost and Recovery Performance}
%From Lemma \ref{lemma:probability of bin}, it is easy to derive that non-iterative sFFT-DT obtains perfect recovery with probability $ \rho=1-Pr(d,a_{m})$ at least.
According to computational cost in Lemma \ref{lemma:cost of non-iterative sFFT-DT} and probability for perfect recovery in Lemma \ref{lemma:probability of bin}, we have Theorem \ref{theorem:probability of sfftdt v1}.
% shows that sFFT-DT perfectly recovers $\bm{\hat{x}}$ along with the computational cost.
\begin{theorem}
\label{theorem:probability of sfftdt v1}
If non-zero frequencies of $\bm{\hat{x}}$ distribute uniformly, given $a_{m}$ and $d$, sFFT-DT perfectly recovers $\bm{\hat{x}}$ with the probability at least $\rho = 1-\frac{N}{d} (\frac{deK}{N(a_{m}+1)})^{a_{m}+1}$ and the computational cost $O(a_{m}\frac{N}{d}\log \frac{N}{d})$ for $a_{m}\leq 4$ and $O(a_{m}\frac{N}{d}\log \frac{N}{d}+a_{m}\frac{N}{d}(\log \log N)^{O(1)})$ for $a_{m} > 4$.
\end{theorem}

%\begin{proof}
%The probability can be derived via $1-Pr(d,a_{m})$.  For $a_{m} \leq 4$, the computational cost contains P1 cost ($O(2a_{m}\frac{N}{d}\log \frac{N}{d})$) and P2 cost ($O(\frac{N}{d}a_{m}^2)$) d; we can replace $a_{m}$ by 4 and the total cost is bounded by $O(\frac{N}{d}\log \frac{N}{d})$. For $a>4$, the cost has been illustrated in Sec. \ref{ssec:MPT}.
%\end{proof}

Based on different parameter settings in Theorem \ref{theorem:probability of sfftdt v1}, we can further distinguish our sFFT-DT from Ghazi {\em et al.}'s sFFT \cite{Ghazi2013} in terms of recovery performance and computational cost as follows.
%, and whether requiring root finding algorithm are decided. Then, we compare two kinds of settings; \\
(sFFT-DT): Set $a_{m}=4$ and $d=O(\frac{N}{K})$.
We have the probability of perfect recovery, $\rho= 1-O(K)O(\frac{e}{5})^{5}$, and computational cost, $O(K \log K)$.\\
(Ghazi {\em et al.}'s sFFT): Set $a_{m}={C\log K}$ and $d = O(\frac{N\log K}{K})$.
We have $\rho = 1 - O(\frac{1}{K}^{0.5C\log C})$ and computational cost $O(K \log K + K(\log\log N)^{O(1)})$. \\
%The former is used in sFFT-DT and the later is used in Ghazi's sFFT $\cite{Ghazi2013}$.
Furthermore, Ghazi {\em et al.}'s sFFT aims to maximize the performance without the constraint of $a_{m} \leq 4$.
Thus, it requires to use an extra root finding algorithm \cite{Pan2002} with complexity being related to the signal length $N$.

On the contrary, sFFT-DT achieves the ideal computational cost, which is independent of $N$, but with the lower bound of successful probability degrading to $0$ for large $K$.
Under this circumstance, sFFT-DT is seemingly unstable.
Nevertheless, if we consider the recovery performance in terms of energy, sFFT-DT can guarantee that most of frequencies are estimated correctly, as Theorem \ref{theorem:probability of sfftdt v2} indicates.
To prove this, we first define some parameters here.
Let $d=\frac{N}{\mu K}$, where $\mu \in \mathbb{N}$ is the user-defined parameter, and let $\tau \in (0,1]$ with $(\frac{\tau}{\mu}-\frac{1}{\mu K})\times 100\%$ representing the proportion of frequencies that cannot be successfully recovered.
% qualifies the trade-off between how many frequencies fail to reconstruct and the failure probability.

\begin{theorem}
\label{theorem:probability of sfftdt v2}
If non-zero frequencies of $\bm{\hat{x}}$ distribute uniformly, given $a_{m}$ and $d=\frac{N}{\mu K}$, sFFT-DT recovers at least $(1-(\frac{\tau}{\mu}-\frac{1}{\mu K}))N$ frequencies of $\bm{\hat{x}}$ with the probability at least $\rho = 1- \left(\frac{d^{a_{m}}K^{a_{m}}  e^{a_{m}+2}}{\tau N^{a_{m}} (a_{m}+1)^{a_{m}+1}}\right)^{\tau K}$, and computational cost $O(a_{m}\frac{N}{d}\log \frac{N}{d})$ for $a_{m}\leq 4$ and $O(a_{m}\frac{N}{d}\log \frac{N}{d}+a_{m}\frac{N}{d}(\log \log N)^{O(1)})$ for $a_{m} > 4$.
\end{theorem}
\begin{proof}
We extend $Pr(d,a_{m})$ derived in Lemma \ref{lemma:probability of bin} as $Pr(d,a_{m},f)$ to represent the probability that at least $f$ frequencies with $a > a_{m}$ is derived as:
\begin{equation}
\begin{aligned}
Pr(d,a_{m},f) & \leq \binom{\frac{N}{d}}{f}  \left(\binom{d}{a_{m}+1} (\frac{K}{N})^{a_{m+1}}\right)^{f} \\
  & \leq \frac{1}{f!}(\frac{N}{d})^{f} \left( (\frac{deK}{N(a_{m}+1)})^{(a_{m}+1)}\right)^{f} \\
  & \leq \left(\frac{K}{f}(\frac{dK}{N})^{a_{m}} \frac{e^{a_{m}+2}}{(a_{m}+1)^{a_{m}+1}}\right)^{f}.
\end{aligned}
\label{Eq: tauK entries}
\end{equation}
Let $f = \tau K$ and plug it in Eq. (\ref{Eq: tauK entries}).
We obtain the result that at least $f$ frequencies of $\bm{\hat{x}}_{d}$ cannot be solved with probability $Pr(d,a_{m},f) \leq \left(\frac{d^{a_{m}}K^{a_{m}}  e^{a_{m}+2}}{\tau N^{a_{m}} (a_{m}+1)^{a_{m}+1}}\right)^{\tau K} $.
In other words, there are at most $(f-1)d=(\frac{\tau}{\mu}-\frac{1}{\mu K})N$ frequencies of $\bm{\hat{x}}$ that cannot be solved with probability $\rho = 1 - Pr(d,a_{m},f)$.
We complete this proof.
%\textcolor{red}{To avoid the term $O(K)$ degrading the probability such as Theorem \ref{theorem:probability of sfftdt v1}, let $f = \tau K$ into Eq. (\ref{Eq: tauK entries}), and it results in at least $f$ frequencies of $\bm{\hat{x}}_{d}$ failing to solve with the failure probability $Pr(d,a_{m},f) \leq   \left(\frac{d^{a_{m}}K^{a_{m}}  e^{a_{m}+2}}{\tau N^{a_{m}} (a_{m}+1)^{a_{m}+1}}\right)^{\tau K} $. In other words,  there are at most $(f-1)d=(\frac{\tau}{\mu}-\frac{1}{\mu K})N$ frequencies of $\bm{\hat{x}}$ failing to solve with the probability $\rho = 1 - Pr(d,a_{m},f)$. We complete this proof.}
\end{proof}
By choosing appropriate $\mu$ and $\tau$, sFFT-DT performs better with successful probability converging to $1$ when $K$ increases, implying that it can work for $K=O(N)$.
For example, by setting $a_{m}=4$ and $d=\frac{N}{\mu K}$, where $\mu$ is $4$, we have $\rho \approx 1-\left(\frac{5\times 10^{-4}}{\tau}\right)^{\tau K}$.
In this case, let $\tau=10^{-2}$ and it means that sFFT-DT correctly recovers at least $99.0\%$ frequencies with probability at least $\delta=1-(\frac{1}{20})^{\tau K}$, which converges to $1$ when $\tau K$ is large enough.

In addition, we further analyze the practical cost of additions and multiplications in detail along with the Big-O constants of computational complexity and find that the Big-O constants in Ghazi's sFFT are larger than those in sFFT-DT.
More specifically, recall that the computational cost of sFFT-DT is composed of two parts:
performing FFTs for obtaining syndromes (P1 cost) and solving Steps (i), (ii), and (iii) of syndrome decoding (P2 cost).
Since $d = \frac{N}{4K}$ was set in our simulations, the Big-O constants for FFT are $96$ for addition and $64$ for multiplication\footnote{Recall that the P1 cost is $2a_{m}\frac{N}{d}\log \frac{N}{d}$. Under the situation that $a_{m}$ is $4$ and $\frac{N}{d}= 4K$, the Big-O constant is $ 2*4*4*3 = 93 $, where $3$ comes from the constant of additions of FFT \cite{Saidi1994}.}.
Since the P2 cost in sFFT-DT is relatively smaller than the P1 cost, it is ignored.

In contrast to sFFT-DT, the Big-O constants of the P1 cost in Ghazi's sFFT \cite{Ghazi2013} are about $6C$ for addition and $4C$ for multiplication ($C$ must be larger than or equal to $2$; otherwise Ghazi {\em et al.}'s sFFT cannot work).
Nevertheless, the Big-O constants of one of the Steps (i) and (iii) within the P2 cost need about $96$ for addition and $160$ for multiplication (the detailed cost analysis is based on \cite{NChen2008}). Even though we do not take Step (ii) into account due to the lack of detailed analysis, the Big-O constants for multiplication in sFFT-DT are far smaller than those of Ghazi {\em et al.}'s sFFT, especially for multiplications.
In addition, for hardware implementation, Ghazi {\em et al.}'s sFFT is more complex than sFFT-DT (due to its analytical solution) because an extra numerical procedure for root finding is required and the computational cost involves $N$.
%An analytical solution in sFFT-DT is easier to implement.
We conclude that there are two main advantages in sFFT-DT, compared to Ghazi's sFFT \cite{Ghazi2013}.
First, the Big-O constants of sFFT-DT are smaller than those of Ghazi {\em et al.}'s sFFT.
Second, our analytical solution is hardware-friendly in terms of implementation.

On the other hand, when the signal is not so sparse with $K$ approaching $N$ ({\em e.g.}, $  K=\frac{N}{8}$ and $d=O(\frac{N}{K})$), the cost of $8$ FFTs in a downsampled signal is almost equivalent to that of one FFT in the original signal.
To further reduce the cost, a top-down iterative strategy is proposed in Sec. \ref{ssec:ItMethod}.

It also should be noted that the above discussions (and prior works) are based on the assumption that $K$ is known.
In practice, $K$ is unknown in advance.
Unfortunately, how to automatically determine K is ignored in the literature.
%To control aliasing, how to decide the downsampling factor depending on $K$  is a challenge.
Instead of skipping this problem, in this paper, we present a simple but effective strategy in Sec. \ref{ssec:how to decide K} to address this issue.

\subsection{Top-Down Iterative Strategy for Iterative sFFT-DT}\label{ssec:ItMethod}
In this section, an iterative strategy is proposed to solve the aliasing problem with an iterative increase of the downsampling factor $d$ according to our empirical observations that the probability of aliasing decreases
fast with the increase of $a$ and the fact that when $d$ is increased, $a$ is increased as well.
The idea is to solve downsampled frequencies from $a=1$ to $a=a_{m}$ iteratively.
During each iteration, the solved frequencies are subtracted from $\hat{x}_{d}$ to make $\hat{x}_{d}$ more sparse.
Under this circumstance, $d$ subsequently is set to be larger values to reduce computational cost without sacrificing the recovery performance.
Fig \ref{fig:Iterative Pyramid} illustrates such an example.
In Fig. \ref{fig:Iterative Pyramid}(b), if we try to solve all aliasing problems in the first iteration, $4$ FFTs are required, since the maximum value of $a$ is $2$.
On the other hand, if we first solve the downsampled frequencies with $a=1$ (at normalized frequency = $\pi$), it costs $2$ FFTs, as shown in Fig. \ref{fig:Iterative Pyramid}(c).
Since $2$ FFTs are insufficient for solving the aliasing problem completely under $a=2$, extra $2$ FFTs are required to solve a more ``sparse'' signal.

The key is how to calculate the $2$ extra FFTs in the above example with lower cost.
Since a more sparse signal is generated by subtracting the solved frequencies from $\hat{x}_{d}$, $d$ can be set to be larger to further decrease the cost of FFT.
As shown in Fig. \ref{fig:Iterative Pyramid}(d), $2$ extra FFTs can be done quickly with a larger $d$ (=$4$) to solve the downsampled frequency (at normalized frequency = $0\pi$) with $a=2$.
Consequently, $d$ is doubled iteratively in our method and the total cost is dominated by that required at the first iteration.

The proposed method with the top-down iterative strategy is called iterative sFFT-DT.

\subsection{Iterative sFFT-DT: Algorithm for Exactly K-Sparse Signals}\label{ssec:exact_algorithm}

In this section, our method, iterative sFFT-DT, is developed and is depicted in Algorithm \ref{Table:our algorithm}, which is composed of three functions, \textbf{main}, \textbf{SubFreq}, and \textbf{SynDec}.
Basically, iterative sFFT-DT solves downsampled frequencies from $a_{m}=1$ to $4$ with an iterative increase of $d$.
Note that, its variation, non-iterative sFFT-DT, solves all downsampled frequencies with $a_{m} = 4$ and fixed $d$.

At the initialization stage, the sets $S$ and $T$, recording the positions of solved and unsolved frequencies, respectively, are set to be empty.
$a_m = 4$ and $d = \frac{N}{4K}$ are initialized.
The algorithmic steps are explained in detail as follows.
%sFFT-DT is designed based on the idea mentioned above.

Function \textbf{main}, which is executed in a top-down manner by doubling the downsampling factor iteratively, is depicted from Line 1 to Line 16.
In Lines 3-4, the input signal $\bm{x}$ is represented by two shift factors $2l$ and $2l+1$.
Then they are used to perform FFT to obtain $\bm{\hat{x}_{d,2l}}$ and $\bm{\hat{x}_{d,2l+1}}$ in Lines 5-6.
In Line 7, the function \textbf{SubFreq}, depicted between Line 17 and Line 22, is executed to remove frequencies from $\bm{\hat{x}_{d,2l}}$ and $\bm{\hat{x}_{d,2l+1}}$ that were solved in previous iterations.
The goal of function \textbf{SubFreq} is to make the resulting signal more sparse.

Line 9 in function \textbf{main} is used to judge if there are still unsolved frequencies.
In particular, the condition $\hat{x}_{d,l}[k]=0$, initially defined in Eq. (\ref{eq:xd with shift}), may imply:
1) $  \hat{x}[k+j\frac{N}{d}]$'s for all $j \in [0,d-1]$ are zero, meaning that there is no unsolved frequency and
2) $  \hat{x}[k+j\frac{N}{d}]$'s are non-zero but their sum is zero, meaning that there exist unsolved frequencies.
To distinguish both, $|\hat{x}_{d,j}[k]| > 0$ for $j\in [0,2l+1]$ is a sufficient condition.
More specifically, if $a$ is less than or equal to $2l+2$, it is enough to distinguish both by checking whether any one of the $2l+2$ equations is not equal to $0$.
If yes, it implies that at least a frequency grid is non-zero; otherwise, all $  \hat{x}[k+j\frac{N}{d}]$'s are definitely zero.
Moreover, Line 9 is equivalent to checking $2l + 2$ equations at the $l$'th iteration.
At $l=0$, two equations ($ \hat{x}_{d,0}[k] $ and $ \hat{x}_{d,1}[k] $) are verified to ensure that all frequencies with $a\leq2$ are distinguished.
At $l=1$, if $k \in T$, it is confirmed that $  \hat{x}[k+j\frac{N}{d}]$'s are non-zero at the previous iteration.
On the contrary, if $k \notin T$, extra 2 equations ($ \hat{x}_{d,2}[k] $ and $ \hat{x}_{d,3}[k] $) are added to ensure that all frequencies with $a\leq4$ are distinguished.
Thus, at the $l$'th iteration, there are in total $2l+2$ equations checked.
%This condition (Line 9) is equivalent to checking $2l+2$ equations.

In Line 11, the function \textbf{SynDec}, depicted in Lines 23-35 (which was described in detail in Sec. \ref{ssec:MPT}), solves frequencies when aliasing occurs. sFFT-DT iteratively solves downsampled frequencies from $a=1$ to $a\leq a_{m}=4$.
Nevertheless, we do not know $a$'s in advance. For example, it is possible that some downsampled frequencies with $a=4$ are solved in the first three iterations, and these solutions definitely fail.
In this case, the solved locations do not belong to $  S_{k}$ (defined in Sec. \ref{ssec:DFT Property}). On the contrary, if the downsampled frequency is solved correctly, the locations must belong to $S_{k}$.
Thus, by checking whether or not the solution satisfies the condition, $s_{j}\;mod\;d=k$ for all $j \in [0,l]$ (Line 30), we can guarantee that all downsampled frequencies are solved under correct $a$'s.
Finally, the downsampling factor is doubled, as indicated in Line 14, to solve the unsolved frequencies in an iterative manner.
This means that the downsampled signal in the next iteration will become shorter and can be dealt faster than that in the previous iterations.

\begin{algorithm}[!htb]
\fontsize{11pt}{0.9em}\selectfont
\setlength{\abovecaptionskip}{0pt}
\setlength{\belowcaptionskip}{0pt}
\caption{Iterative sFFT-DT for exactly $K$-sparse signals.}
\label{Table:our algorithm}
\begin{tabular}[t]{p{17.7cm}l}
%\hspace*{-150pt}\makebox[\linewidth]{\rule{4cm}{0.4pt}}
\textbf{Input:} $\bm{x}$, $K$;\quad \textbf{Output:} $\bm{\hat{x}}$; \\
\textbf{Initialization:} $\bm{\hat{x}}=\mathbf{0}$, $  d=\frac{N}{4K}$, $S=\{ \}$, $T=\{ \}$, $a_{m}=4$; \\
\hline\hline
01. \textbf{function} \textbf{main}()\\
02. \quad\textbf{for} $l = 0$ to $a_{m} -1$ \\
03. \quad\quad $x_{d,2l}[k]=x[dk+2l]$ for $  k\in [0,\frac{N}{d}-1]$;\\
04. \quad\quad $x_{d,2l+1}[k]=x[dk+2l+1]$ for $  k\in [0,\frac{N}{d}-1]$;\\
05. \quad\quad $\bm{\hat{x}_{d,2l}}=\mathbf{FFT}( \bm{x_{d,2l}} )\times d$;\\
06. \quad\quad $\bm{\hat{x}_{d,2l+1}}=\mathbf{FFT}( \bm{x_{d,2l+1}} )\times d$;\\
07. \quad\quad $\mathbf{SubFreq}(\bm{\hat{x}_{d,2l}},\bm{\hat{x}_{d,2l+1}},\bm{\hat{x}},d,l,S)$;\\
08. \quad\quad \textbf{for} $k = 0$ to $  \frac{N}{d}-1$\\
09. \quad\quad \quad \textbf{if} ($k \in T$ or $|\hat{x}_{d,2l}[k]| > 0$ or $|\hat{x}_{d,2l+1}[k]| > 0$) \\
10. \quad\quad \quad \quad $m_{j}= \hat{x}_{d,j}[k]$ for $j \in [0,2l+1]$; \\
11. \quad\quad \quad \quad $\mathbf{SynDec}(\bm{m},l,d,k,\bm{\hat{x}},S,T)$;\\
12. \quad\quad \quad \textbf{end if}\\
13. \quad\quad \textbf{end for} \\
14. \quad\quad $d=2d$; \\
15. \quad\quad All elements in $T$ modulo $  \frac{N}{d}$.\\
16. \quad\textbf{end for}\\
%17. \textbf{end} \textbf{function}\\
%\\
17. \textbf{function} \textbf{SubFreq} $(\bm{\hat{x}_{d,2l}},\bm{\hat{x}_{d,2l+1}},\bm{\hat{x}},d,l,S)$\\
18. \quad \textbf{for} $k \in S$\\
19. \quad \quad $k_{d}=k $ mod $   \frac{N}{d}$;\\
20. \quad \quad $   \hat{x}_{d,2l}[k_{d}]=\hat{x}_{d,2l}[k_{d}]- \hat{x}[k]e^{\frac{i2\pi k(2l)}{N}} $;\\
21. \quad \quad $   \hat{x}_{d,2l+1}[k_{d}]=\hat{x}_{d,2l+1}[k_{d}]- \hat{x}[k]e^{\frac{i2\pi k(2l+1)}{N}} $;\\
22. \quad \textbf{end for}\\
%24. \textbf{end} \textbf{function}\\
%\\
23. \textbf{function} \textbf{SynDec} $(\bm{m},l,d,k,\hat{x},S,T)$\\
24. \quad \textbf{if} $l=0$\\
25. \quad \quad $  z_{0}= ( \frac{m_{1}}{m_{0}})$;\;\;$p_{0}= m_{0}$;\\
%26. \quad \quad $p_{0}= m_{0}$;\\
26. \quad \textbf{else} \\
27. \quad \quad  Solve the aliasing problem with $a=l+1$ by \\
    \quad \quad \quad \;\;syndrome decoding, described in Sec. \ref{ssec:MPT}. \\
28. \quad \textbf{end if} \\
29. \quad $s_{j}= (\ln z_{j})N/i2\pi $ for all $j \in [0,l] $; \\
30. \quad  \textbf{if} ($s_{j}\;mod\;d$) $ = k$ for all $j \in [0,l] $\\
31. \quad \quad $S=S\cup \bm{s}$;\\
32. \quad \quad $\hat{x}[s_{j}]=p_{j}$ for all $j \in [0,l]$;\\
33. \quad \textbf{else} \\
34. \quad \quad  $T=T\cup \bm{s}$;\\
35. \quad  \textbf{end if}\\
%40. \textbf{end} \textbf{function}\\
\hline
\end{tabular}
\end{algorithm}
\normalsize

\subsection{Performance and Computational Complexity of Iterative sFFT-DT}\label{ssec:conv_analysis}
%In this section, the performance and computation complexity of our method are discussed.
We first discuss the complexity of iterative sFFT-DT.
The cost of the outer loop in function \textbf{main} (Steps 5 and 6) is bounded by two FFTs.
As mentioned in Theorem \ref{theorem:probability of sfftdt v2}, $d$ is set to be $  \frac{N}{4K}$, the dimensions of $x_{d,2l}$ and $x_{d,2l+1}$ are $O(K)$, and FFT costs $O(K \log K)$ in the first iteration.
Since $d$ is doubled iteratively, the total cost of  $a_m$ iterations is still bounded by $O(K \log K)$.
In addition, the function \textbf{SubFreq} costs $O(K)$ operations due to $|S|\leq K$.

The inner loop of the function \textbf{main} totally runs $O(K)$ times, which is not related to the outer loop, since at most $K$ frequencies must be solved.
The cost at each iteration is bounded by the function \textbf{SynDec}.
Recall that the P2 cost, as described in Sec. \ref{ssec:MPT}, requires $  O(\frac{N}{d}a^{2})$.
%Though step (ii) has no analytic  solution for $l > 3$, it can be solved in $O(a^2)$ \cite{NChen2008} because the roots must belong to the finite set, $\{e^{i2\pi(k+\frac{N}{d})l/N} \ | \ l \in [0,d-1] \}$.
More specifically, since $d$ is doubled iteratively, $  \frac{N}{d}$ can be derived to depend on $  \frac{O(K)}{2^{l}}$ from the initial setting $  d=O(\frac{N}{K})$.
Therefore, \textbf{SynDec} at the $l$'th iteration costs $  O(\frac{K}{2^l}(l+1)^{2})$ and requires  $  O(\frac{K}{2^0}1^{2}+\frac{K}{2^1}2^{2}+...+\frac{K}{2^{3}}4^{2} ) \leq O(6.25K) = O(K)$ in total.
That is, the inner loop (Steps 8$\sim$13) costs $O(K)$, given an initial downsampling factor of $  d=\frac{N}{4K}$ and $a_{m}=4$.

In sum, the proposed algorithm, iterative sFFT-DT, is dominated by ``FFT'' and costs $O(K \log K)$ operations.
Now, we discuss Big-O constants for operations of addition and multiplication, respectively.
Since $d$ is doubled iteratively, the P1 cost of syndrome decoding gradually is reduced in the later iterations.
The total cost is $  \sum_{i=1}^{a_m} O(2\frac{N}{2^{i-1}d}\log\frac{N}{2^{i-1}d})$, where $a_m=4$.
Due to the fact that iterative sFFT-DT possibly recovers $\bm{\hat{x}}$ with less than $a_m=4$ iterations, the benefit in reducing the computational cost depends on the number of iterations.
In the worst case, the cost is about $  O((2+1+\frac{1}{2}+\frac{1}{4})\frac{N}{d}\log\frac{N}{d})= O(3.75\frac{N}{d}\log\frac{N}{d})$ under $a_m=4$.
Recall that the P1 cost of syndrome decoding in non-iterative sFFT-DT is $  O(2a_m \frac{N}{d}\log\frac{N}{d})$.
With $a_{m}=4$, the Big-O constants in non-iterative sFFT-DT are two times larger than those in iterative sFFT-DT. Similarly, in the best case ({\em i.e.}, $a$'s of all frequencies are $1$), the former is about $4$ times larger than the latter. Thus, it is easy to further infer Big-O constants of iterative sFFT-DT. For instance, since the Big-O constant of addition for non-iterative sFFT-DT is $96$, the Big-O constants for iterative sFFT-DT addition range from $12\times 2=24$ (the best case) to $12 \times 3.75 = 45$ (the worst case) and those for multiplication range from $16$ to $30$.

As for recovery performance in iterative sFFT-DT, since the downsampling factor $d$ is doubled along with the increase of iterations, a question, which naturally arises, is if a larger downsampling factor leads to more new aliasing artifacts.
If yes, these newly generated collisions possibly degrade the performance of iterative sFFT-DT.
If no, the iterative style is good since it reduces computational cost and maintains recovery performance.

In Lemma \ref{lemma:probabilityofaliasing}, we prove that the probability of producing new aliasing artifacts after a sufficient number of iterations will approach zero.

\begin{lemma}
\label{lemma:probabilityofaliasing}
Suppose $K$ non-zero entries of ${\bm\hat{x}}$ distribute uniformly ({$i.e.$}, with probability $  \frac{K}{N}$).
Let $Pr_{l}^{ali}$ be the probability that new aliasing artifacts are produced at the $l$'th iteration in iterative sFFT-DT. Let $K_{l}$ be the number of frequencies with $a \geq  l+1$ at the $l$'th iteration ($0\leq l\leq 3$, $K_{0}=K$). If $  K_{l} \leq \frac{K}{2^{l}}$, we have $Pr_{l}^{ali} < \frac{1}{2^{l+1}}(\frac{d}{N})K^{2}$.
\end{lemma}

\begin{proof}
According to Algorithm \ref{Table:our algorithm}, after the first iteration ($l=0$), all downsampled frequencies with only $a=1$ aliasing term are solved.
Thus, we focus on discussing the probability of producing new aliasing artifacts under $l \geq 1$.
By the same idea of Lemma \ref{lemma:probability of bin}, we can define $  \frac{N}{d}\binom{d}{2}(\frac{K}{N})^{2}$ to be the probability that there is a downsampled frequency with $a \geq 2$ aliasing terms.
In the second iteration. $l=1$, however, some non-zero frequencies have been solved in previous iterations.
Thus, the number of remaining non-zero frequencies are no longer $K$ and $\frac{K}{N}$ and should be modified.
%At $l=1$, each downsampled frequencies must contains 0 or more than 1 non-zero frequencies of $\bm{\hat{x}}$.
In other words, the number of downsampled frequencies with $a \geq 2$ must be less than $K_{l}$ for $l\geq 1$ and $  \frac{N}{d}\binom{d}{2}(\frac{K_{l}}{N})^{2}$ becomes the upper bound of the probability that there exists a downsampled frequency of producing new aliasing artifacts.

According to our iterative sFFT-DT algorithm, let $d_{l}=2^ld$ for $l\geq 1$.
We can derive:
\small
\begin{equation}
  Pr_{l}^{ali} \leq \frac{N}{d_{l}} \binom{d_{l}}{2} (\frac{K_{l}}{N})^{2} \leq \frac{N}{2^{l}d}(\frac{2^{l}d}{N})^{2}\frac{K^{2}_{l}}{(2)!}\leq 2^{l-1}(\frac{d}{N})K^{2}_{l}.
\label{Eq: Bound of aliasing}
\end{equation}
\normalsize
Eq. (\ref{Eq: Bound of aliasing}) converges to $0$ when $  K_{l}\leq\frac{K}{2^{l}}$. By initializing $d$ properly, almost $K$ frequencies can be solved in the first few iterations. This makes $  K_{l}\leq\frac{K}{2^{l}}$ easy to be satisfied.
Under this circumstance, $d=\frac{N}{4K}$ would be a good choice. By replacing $K_{l}$ with $  \frac{K}{2^{l}}$, we can derive $Pr_{l}^{ali} \leq \frac{1}{2^{l+1}}(\frac{d}{N})K^{2}$.
When $l$ increases to be large enough, the probability of $Pr_{l}^{ali}$ will be small since $  \frac{1}{2^{l+1}}(\frac{d}{N})K^{2} \rightarrow 0$. \end{proof}

Lemma \ref{lemma:probabilityofaliasing} indicates the probability of producing new aliasing artifacts in an asymptotic manner.
This provides us the information that the probability of producing new aliasing finally converges to zero.
In our simulations, we actually observe that the exact probability with new aliasing is very low under $d=\frac{N}{4K}$, implying that the iterative approach can reduce the computational cost and maintain the recovery performance effectively.
%The outer loop of \textbf{main} runs $t$ times.
%Given $d=O(\frac{N}{K})$ and $N^{+}\in [2^{1}\ 2^{2}]$, if non-zero frequencies are distributed uniformly, then $t$ is set to be $\log K$ according to Lemma \ref{lemma:probabilityofaliasing} below.

\subsection{A Simple Strategy for Estimating Unknown Sparsity $K$ }\label{ssec:how to decide K}
As previously described, the sparsity $K$ of a signal is important in deciding the downsampling factor $d$.
Nevertheless, $K$ is, in general, unknown.
In this section, we provide a simple bottom-up strategy to address this issue.

First, we set a large downsampling factor $d=N$, and then run sFFT-DT.
If there is any downsampled frequency that cannot be solved, then $d$ is halved and sFFT-DT is applied to solve $\bm{\hat{x}}$ again.
When $d$ is halved iteratively until the condition in either Theorem \ref{theorem:probability of sfftdt v1} or Theorem \ref{theorem:probability of sfftdt v2} is satisfied, sFFT-DT guarantees one to stop with the probability indicated in either Theorem \ref{theorem:probability of sfftdt v1} or Theorem \ref{theorem:probability of sfftdt v2}.
This strategy needs the same computational complexity required in sFFT-DT with known $K$ because the cost with $d=N$ is $  O(2a_{m}\frac{N}{d}\log\frac{N}{d}) = O(2a_{m})$ and the total cost is $O(2a_{m})+O(2a_{m}2\log 2)+...+ O(2a_{m}K \log K) < O(4a_{m}K \log K)$.
Thus, sFFT-DT with the strategy of automatically determining $K$ costs double the one with known $K$.

\subsection{Simulation Results for Exactly K-Sparse Signals}\label{Sec: Experimental Results}
Our method\footnote{Our code is now available in http://www.iis.sinica.edu.tw/pages/lcs/publications\_en.html (by searching ``Others'').}, iterative  sFFT-DT, was verified and compared with FFTW (using the plan of FFTW\_ESTIMATE  (http://www.fftw.org/)), sFFT-v3 \cite{Haitham2012} (its code was downloaded from http://spiral.net/software/sfft.html), GFFT (using the plan of GFFT-Fast-Rand, which is an implementation of \cite{Iwen2010} and is discussed in \cite{Segal2012} in detail (its code was downloaded from http://sourceforge.net/projects/gopherfft/)), and Ghazi {\em et al.}'s sFFT \cite{Ghazi2013} for exactly $K$-sparse signals.
The simulations for sFFT-DT, FFTW, GFFT, and Ghazi {\em et al.}'s sFFT were conducted with an Intel CPU Q6600 and $2.99$ GB RAM under Win 7.
sFFT-v3 was run in Linux because the source code was released in Linux's platform.
The signal $\bm{x}$ in time domain was produced as follows:
1) Generate a $K$-sparse signal $\bm{\hat{x}}_{ori}$ and
2) $\bm{x}$ is obtained by inverse FFT of $\bm{\hat{x}_{ori}}$.
%The approximation error is defined as $\frac{\| \bm{\hat{x}} - \bm{\hat{x}_{ori}} \|_{1}}{\| \bm{\hat{x}_{ori}} \|_{1}}$.

For sFFT-DT, the initial $d$ is set according to $  d=\frac{N}{4K}$, based on Theorem \ref{theorem:probability of sfftdt v2}.
For sFFT-v3, $d$ was automatically assigned, according to the source code.
For Ghazi {\em et al.}'s sFFT, $  d=\frac{N}{K} 2^{\lfloor \log\log K \rfloor}$ and $a_{m}=2\log K$, where $2^{\lfloor \log\log K \rfloor}$ is involved to enforce $  \frac{N}{d}$ being an integer.
If $\log\log K$ is an integer, $  d=\frac{N}{K} 2^{\lfloor \log\log K \rfloor}= \frac{N\log K}{K}$.

The comparison of computational time is illustrated in Fig. \ref{fig:Computation Time}.
Fig. \ref{fig:Computation Time}(a) shows the results of computational time versus sparsity under $N=2^{24}$.
For $  K\leq\frac{N}{2^{4}}$, our algorithm outperforms FFTW.
Moreover, sFFT-v3 \cite{Haitham2012}\cite{Haitham2012_1} is only faster than FFTW when $  K\leq\frac{N}{2^{6}}$ and is comparable to Ghazi {\em et al.}'s sFFT.
We can also observe from Fig. \ref{fig:Computation Time}(a) that Ghazi {\em et al.}'s sFFT is slower than iterative sFFT-DT because the P2 cost of syndrome decoding in Ghazi {\em et al.}'s sFFT dominates the computation.
%If more approximation errors are allowed, our algorithm saves more computation time.
Compared to sFFT-v3 and Ghazi {\em et al.}'s sFFT, our method, iterative sFFT-DT, is able to deal with FFT of signals with large $K$.
GFFT demonstrates the worst results as it crashes when $K>2^{12}$ under $N=2^{24}$.
%GFFT in Fig. \ref{fig:Computation Time}(a) only show the computational cost until $K \leq 2^{12}$.  }
Fig. \ref{fig:Computation Time}(b) shows the results of computational time versus signal dimension under fixed $K$.
It is observed that the computational time of iterative sFFT and Ghazi {\em et al.}'s sFFT is invariant to $N$, but our method is the fastest.
%But, sFFT-v3 is slightly related to $N$.

\begin{figure*}[!t]
\begin{minipage}[b]{.48\linewidth}
%  \centering
\centering{\epsfig{figure=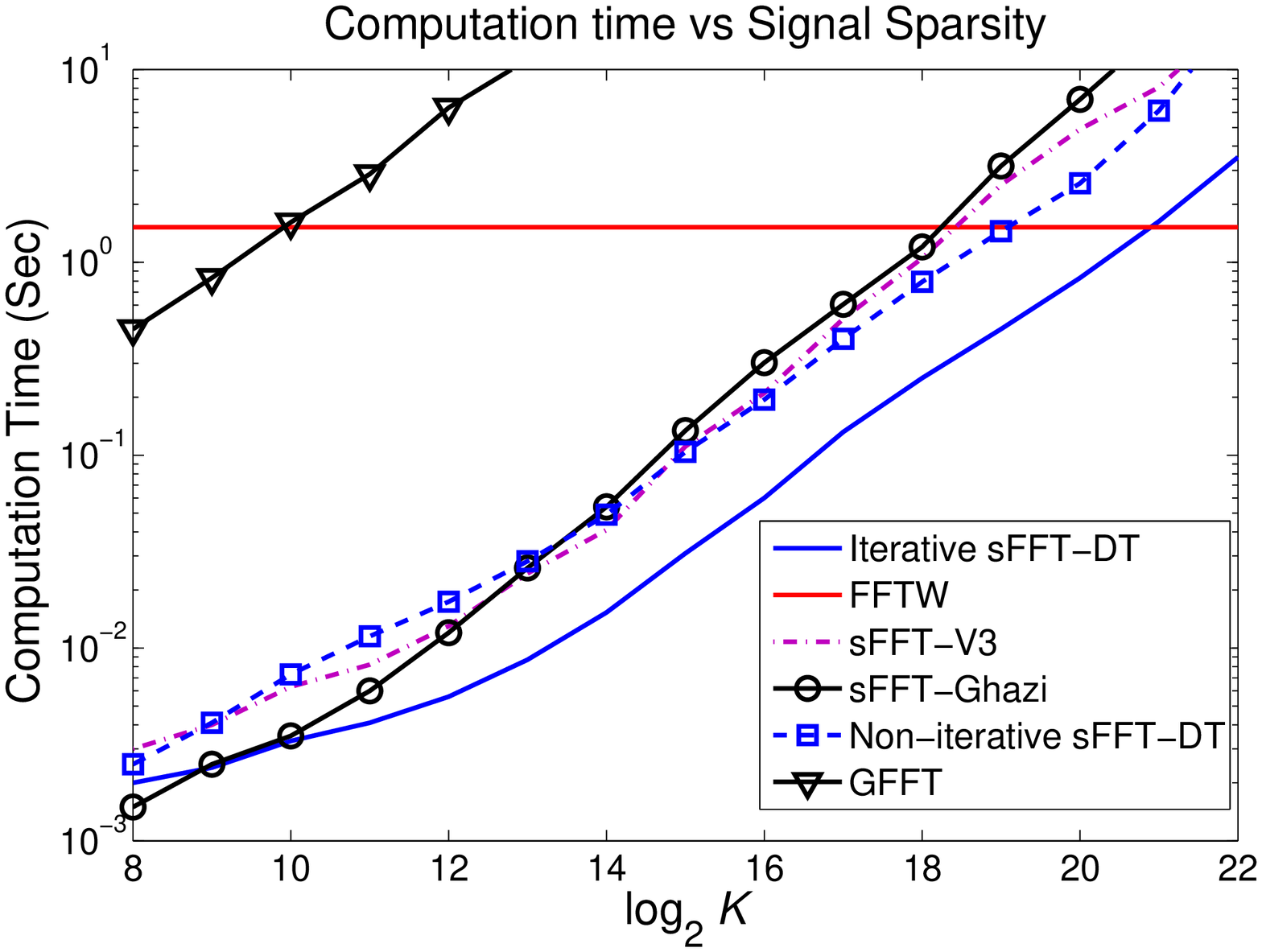,width=3.55in}}
  \centerline{(a)}
\end{minipage}
\begin{minipage}[b]{.48\linewidth}
%  \centering
  \centering{\epsfig{figure=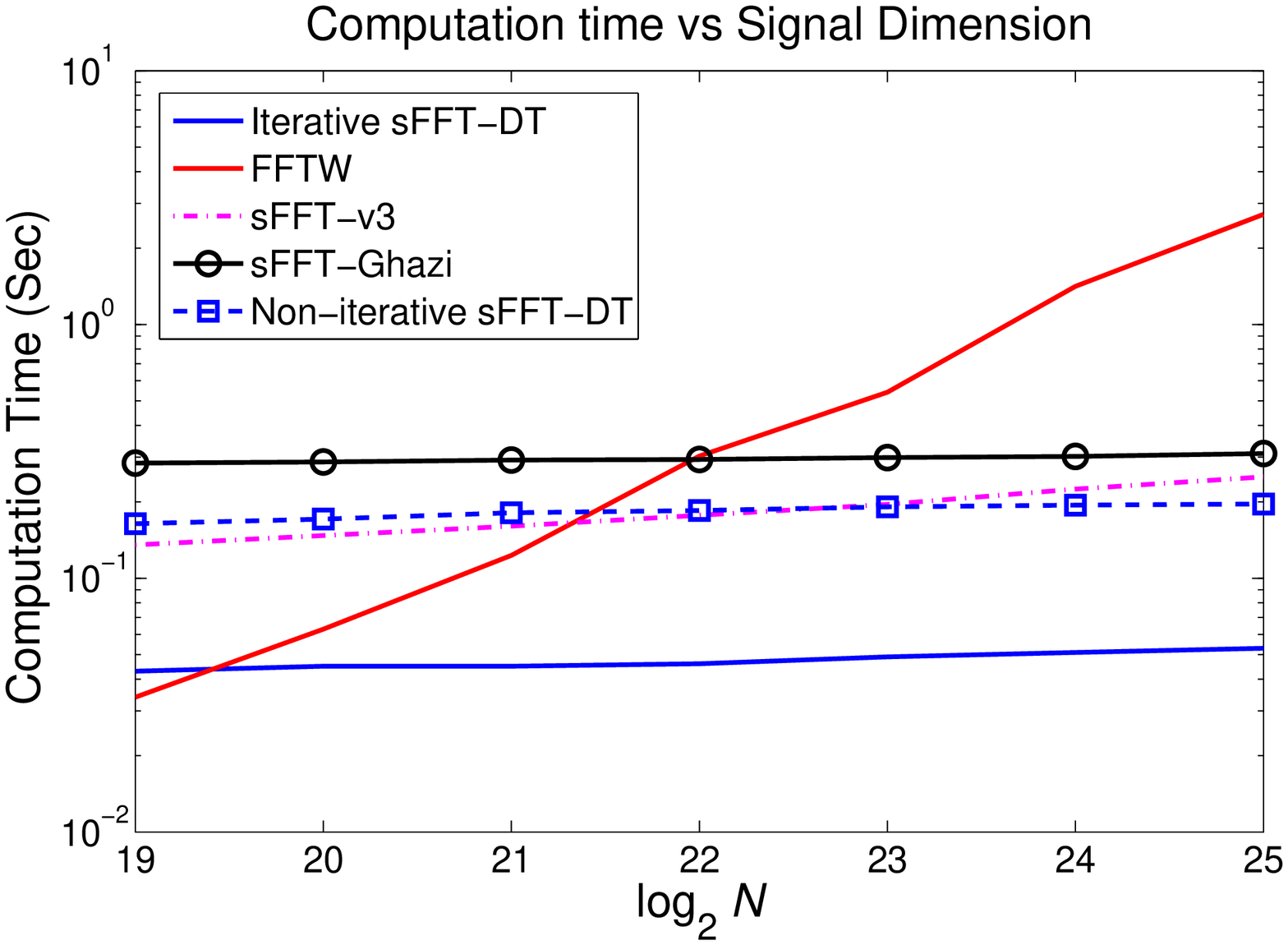,width=3.55in}}
  \centerline{(b)}
\end{minipage}
\hfill
\caption{Comparison of computational time for exact K-sparse signals. (a) Computational time vs. sparsity under $N=2^{24}$. (b) Computational time vs. signal dimension under $K=2^{16}$ and $a_{m}=4$.}
\label{fig:Computation Time}
\end{figure*}

Moreover, according to Theorem \ref{theorem:probability of sfftdt v1}, the performance of non-iterative sFFT-DT seems to be inferior to that of Ghazi {\em et al.}'s sFFT.
%By comparing Theorem \ref{theorem:probability of sfftdt v1} and Theorem \ref{theorem:probability of sfftdt v2}, iterative sFFT-DT seems to be worse than non-iterative sFFT-DT in terms of recovery performance.
Nevertheless, the successful probability described in Theorem \ref{theorem:probability of sfftdt v1} is merely a lower bound.
In our simulations, we compare the recovery performance among three approaches: non-iterative sFFT-DT, iterative sFFT-DT, and Ghazi {\em et al.}'s sFFT \cite{Ghazi2013}.
The parameters for both proposed approaches and Ghazi {\em et al.}'s sFFT were set based on Theorem \ref{theorem:probability of sfftdt v1}.

We have the following observations from Fig. \ref{fig:Performance}, where signal length is $N=2^{20}$.
First, although the theoretical result derived in Theorem \ref{theorem:probability of sfftdt v1} indicates that the performance decreases along with the increase of $K$, it is often better than Ghazi {\em et al.}'s sFFT \cite{Ghazi2013}.
In fact, it is observed that the performance of Ghazi {\em et al.}'s sFFT oscillates.
The oscillation is due to the fact that the floor operation in $2^{\lfloor \log\log K \rfloor}$ (from $  d=\frac{N}{K} 2^{\lfloor \log\log K \rfloor}$) acts like a discontinuous function and leads to large variations of setting $d$.
The recovery performance would benefit by setting small $d$ at the expense of requiring greater computational cost.
%For example, $d=32768$ when $K=2^{7}$.
%But if the floor operation in $2^{\lfloor \log\log K \rfloor}$ is canceled, we have $d=57344$.
%In other words, the recovery performance will be beneficial from setting $d$ smaller than its theoretical value.
%When $2^{\lfloor \log\log K \rfloor}$ is an integer ({\em i.e.} $K=2^{8}$), however, there is no extra benefit and thus the performance is degraded and returns to theoretical expectation.
Second, iterative sFFT-DT degrades the recovery performance gradually as $K$ increases while, at the same time, the number of collisions ($0 \leq a \leq d$) decreases as well.
That is the reason the performance returns to $100\%$ when $K=2^{16}$ under the case that $  d = \frac{N}{4K}$.

\begin{figure}[t]
\begin{minipage}[b]{.98\linewidth}
%  \centering
\centering{\epsfig{figure=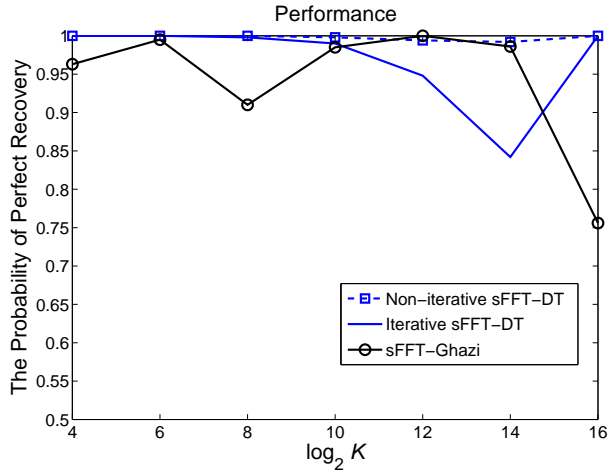,width=3.55in}}
\end{minipage}
\hfill
\caption{Recovery performance comparison among non-iterative sFFT-DT, iterative sFFT-DT, and Ghazi {\em et al.}'s sFFT \cite{Ghazi2013} for exact $K$-sparse signal. The signal length is $N=2^{20}$.}
\label{fig:Performance}
\end{figure}

\section{(Non-Iterative) sFFT-DT for Generally $K$-Sparse Signals}\label{Sec: sFFT-DT: General K-Sparse}
For sparse FFT of a generally $K$-sparse signal $\bm{x}$, the goal is to compute an approximate transform $\bm{\hat{x}}_{out}$ satisfying:
\small
\begin{equation}
\bm{\hat{x}}_{out}= \arg \min_{\bm{\hat{x}}^{'}}  \| \bm{\hat{x}}-\bm{\hat{x}}^{'} \|_{2},
\label{eq:general K-sparse signal}
\end{equation}
\normalsize
where $\bm{\hat{x}}_{out}$ is exactly $K$-sparse and $\bm{\hat{x}}$ is generally $K$-sparse. Without loss of generality, we assume that all frequencies in $\bm{\hat{x}}$ are non-zero.
Similar to exactly $K$-sparse signals, we assume that $K$ significant frequencies (with the first $K$ largest magnitudes) of $\bm{\hat{x}}$ distribute uniformly.

Due to generally $K$-sparsity of $\bm{\hat{x}}$, the right-hand side of Eq. (\ref{eq:xd}) will contain $d$ terms.
When solving syndrome decoding, the remaining insignificant terms will perturb the coefficients of polynomial in Eq. (\ref{eq:poly fun of MPT}).
In addition, how to estimate the roots for perturbed polynomial is an ill-conditioned problem ({\em i.e.}, Wilkinson's polynomial \cite{Wilkinson1959}).
Thus, instead of directly estimating roots by syndrome decoding, we reformulate the aliasing problem in terms of an emerging methodology, called Compressive Sensing (CS) \cite{Donoho2006}\cite{Candes2008}, that has been received much attention recently.
%\textcolor{red}{, and reformulate the aliasing problem as CS problem.}

Compressive sensing (CS) is originally proposed for sampling signals under the well-known Nyquist rate.
If the signal follows the assumption that it is sparse in some transformed domain, CS shows that the signal can be recovered from fewer samples, even though the signal is interfered by noises.
The model of CS is formulated as:
\begin{equation}
\bm{y}=\Phi (\bm{s}+\bm{\eta})+\bm{e},
\label{Eq: Random Projection}
\end{equation}
where $\bm{s}$ is a sparse signal, $\Phi$ is a sensing matrix, $\bm{y}$ is the samples (also called measurements), $\bm{\eta}$ is a signal noise, and $\bm{e}$ is a measurement noise.
It should be noted that $\Phi$ must satisfy either the restricted isometry property (RIP) \cite{Candes2005}\cite{Gan2009} or mutual incoherence property (MIP) \cite{Donoho2001}\cite{Welch1974} for successful recovery with high probability.
It has been shown that Gaussian random matrix and partial Fourier matrix \cite{Candes2006-F} are good candidates to be $\Phi$.

For sFFT-DT of a generally $K$-sparse signal, we formulate the aliasing problem as the CS problem shown in Eq. (\ref{Eq: Random Projection}).
The strategy based on CS is motivated by the following facts:
1) The magnitudes of significant terms must be larger than those of insignificant terms.
2) The number of significant terms is less than that of insignificant terms.
Thus, estimating the locations and values of significant terms is
%\textcolor{red}{\sout{reformulated as a sparse signal recovery problem}
consistent with the basic assumption in the context of CS.

Unlike iterative sFFT-DT for exactly $K$-sparse signals, the iterative approach cannot work for generally $K$-sparse signals since one cannot guarantee that an exact solution can be attained at each iteration without propagating recovery errors for subsequent iterations.
Therefore, we only study non-iterative sFFT-DT for generally $K$-sparse signals.

In this section, how to formulate the aliasing problem as the CS problem is described in Sec. \ref{Sec: Refinement}.
We discuss CS-based performance along with the sufficient conditions for CS successful recovery in Sec. \ref{Sec: Analysis CS}.
In Sec. \ref{Sec: decide a}, we propose a pruning strategy along with proofs to improve the recovery performance and reduce the computational cost.
The detailed algorithm is described in Sec. \ref{Sec: algorithm for general K sparse}.
The computational time analysis and simulations are, respectively, described in Sec. \ref{Sec: Complexity  for General Ksparse Signal} and Sec. \ref{Sec: experimental result of General}.

\subsection{Problem Formulation}\label{Sec: Refinement}
Recall the BCH codes in Eq. (\ref{eq:MPT obejective fun}), where the locations ($z_j^l$) and values ($p_j$) as variables.
Since all candidate locations are known and belong to $S_{k}=\{ k , \  k+\frac{N}{d} \  , ...\  ,  k+(d-1)\frac{N}{d}\}$, instead of considering both $z_j^l$ and $p_j$ as variables in Eq. (\ref{eq:MPT obejective fun}), only $p_{j}$ are thought of as unknown variables here.
Then, we reformulate the aliasing problem in terms of the CS model as:
%$\bm{m}=\bm{W}\bm{b}$
\small
\begin{equation}
%\scriptsize
{\underbrace{\begin{bmatrix}
       m_{n_{0}}\\
       m_{n_{1}}\\
       \vdots \\
       m_{n_{r-1}}
\end{bmatrix}}_{\bm{y}}
= \underbrace{\begin{bmatrix}
       F_{k,n_{0}}^{-1}  &  \cdots   &   F_{k+(d-1)\frac{N}{d},n_{0}}^{-1}    \\
       F_{k,n_{1}}^{-1}          & \cdots  &  F_{k+(d-1)\frac{N}{d},n_{1}}^{-1} \\
        \vdots            & \ddots  &\\
       F_{k,n_{r-1}}^{-1}            & \cdots   &F_{k+(d-1)\frac{N}{d},n_{r-1}}^{-1}
\end{bmatrix}
}_{\Phi}
\underbrace{
\begin{bmatrix}
       b_{0}     \\
       b_{1}        \\
        \vdots\\
       b_{d-1}
\end{bmatrix}}_{\bm{s}+\bm{\eta}}},
\label{eq:CS Formulation}
\end{equation}
\normalsize
where $b_{t}$ is the value at ($k+t\frac{N}{d}$)'th frequency for $t \in [0,d-1]$ and $n_{j}$ is the shift factor for $j \in [0,r-1]$.
Let the left-hand side in Eq. (\ref{eq:CS Formulation}) be $\bm{y}$ as in Eq. (\ref{Eq: Random Projection}) and let the right-hand side be $\Phi (\bm{s}+\bm{\eta})$ with $\bm{e}=\bm{0}$.
It should be noted that $\bm{b} = [b_{0},b_{1},...,b_{d-1}]$ is composed of $\bm{s}$ (significant terms with $a$ non-zero frequencies) and $\bm{\eta}$ (insignificant terms with $d-a$ non-zero frequencies).
Therefore, $\bm{y} \in \mathbb{C}^{r}$, $\Phi \in \mathbb{C}^{r \times d}$, and $\bm{s}$ and $\bm{\eta} \in \mathbb{C}^{d}$.
In fact, Eq. (\ref{eq:CS Formulation}) can degenerate to Eq. (\ref{eq:MPT obejective fun}).
For example, Eq. (\ref{eq:MPT obejective fun}) is expressed as $\bm{m}=\bm{Z}\bm{p}$, where $\bm{p}=[p_{0},p_{1},...,p_{a-1} ]^{T}$ and $\bm{Z}$ is a matrix in which $(i,j)$'th entry is $z^{i}_{j}$.
If $\Phi_{p}$ is the matrix by pruning the columns of $\Phi$ corresponding to insignificant terms, then $\bm{Z}=\Phi_{p}$.

To solve $\bm{s}$ given $\bm{y}$, Eq. (\ref{eq:CS Formulation}) have infinite solutions since $r < d$.
Conventionally, two strategies \cite{Donoho2006}, $\ell_{1}$-minimization and greedy approaches, are popularly used for sparse signal recovery in CS.
Among them, Subspace Pursuit (SP) \cite{Dai2009} is one of the greedy algorithms and requires $O(ard)$ for solving Eq. (\ref{eq:CS Formulation}).
SP runs at most $\frac{N}{d}$ times, leading to the total cost of SP being $O(arN)$.
Similar to exactly $K$-sparse signal, we choose $a_{m}$ as the maximum number of collisions for all downsampled frequencies. Thus, the maximum cost of solving SP is $O(a_{m}rN)$.
Since $a_{m}$ can be chosen as a constant to ensure that most of downsampled frequencies satisfy $a \leq a_{m}$ by tuning an appropriate $d$ shown in Theorem \ref{theorem:probability of sfftdt v2}, the cost of SP finally is simplified into $O(rN)$.

The next step is how to set $r$, which is very important and related to computational complexity and recovery performance.
In fact, $r$ is directly related to the sampling rate in CS.
Candes and Wakin \cite{Candes2008} pointed out that $r$  must satisfy $r \geq O(a \log   \frac{d}{a})$ to recover $\bm{s}$ given $\bm{y}$ and $\Phi$.
If $  d=O(\frac{N}{K})$, then $r \geq O(a \log \frac{N}{aK})$.
In other words, $N$ is also a parameter that impacts the size of $r$.
This will make the cost of solving Eq. (\ref{eq:CS Formulation}) related to $N$ and lead to massive computational overhead, which is unacceptable as sFFT-DT must be faster than FFT.
Thus, in sFFT-DT, $r$ is forced to be $3a_{m}$ and the total cost of SP becomes $O(N)$.
In other words, %\textcolor{red}{\sout{the maximum number of syndromes required to solve ???CS problems with random shift operator are $3a_{m}$, where $a_{m}$ denotes the maximum number of significant terms among all of the downsampled frequencies}
we generate at most $r=3a_{m}$ syndromes for solving Eq. (\ref{eq:CS Formulation}).
Since $r$ is fixed, it is expected to degrade performance when $d$ becomes large.
We will discuss the recovery performance in Sec. \ref{Sec: Analysis CS} under this setting.

Finally, we discuss the relationship between shift factors $n_{j}$'s and $\Phi$, in which both affect the recovery performance in CS.
From the theory of CS, the performance also depends on mutual coherence of $\Phi$, which is defined as:
\small
\begin{equation}
\theta = \max_{i,j, i\neq j} \langle \Phi_{i},\Phi_{j}\rangle.
\label{eq:mutual incoherence}
\end{equation}
\normalsize
In this case, the phase difference between $F_{k,n_{j}}^{-1}$ and $F_{k+\frac{N}{d},n_{j}}^{-1}$ is $2\pi\frac{n_{j}}{d}$, as defined in Eq. (\ref{eq:CS Formulation}).
Recall $l \in \{ n_{0},n_{1},...,n_{r-1} \}$ ($0\leq j < r$).
If we set $n_{j}\in [0,2a-1]$, the maximum shift $2a-1$ is encountered and the phase difference between $F_{k,n_{j}}^{-1}$ and $F_{k+\frac{N}{d},n_{j}}^{-1}$ still approaches $0$ with $ 2\pi\frac{2a-1}{d} \rightarrow 0$ ($d \gg a$).
Under this circumstance, $\theta \rightarrow 1$ and perfect sparse recovery will become impossible.
Thus, $n_{j}$'s are uniformly drawn from $[0,d-1]$.
This makes $\Phi$, in fact, be a partial Fourier random matrix and its mutual coherence will be small, as shown in \cite{Candes2006-F}.

In sum, sFFT-DT for generally $K$-sparse signals first performs FFTs of downsampled signals with random shift factors and then for each downsampled frequency, the aliasing problem is reformulated in terms of CS model solved by subspace pursuit.

\subsection{Analysis of CS-based Approach}\label{Sec: Analysis CS}
In this section, we describe the recovery performance and computational cost based on the CS model-based solver, indicated in Eq. (\ref{eq:CS Formulation}).
For subsequent discussions, we let $\bm{\hat{x}}= \bm{\hat{x}_{s}} +\bm{\hat{x}_{ns} }$, where $\bm{\hat{x}_{s}}$ and $\bm{\hat{x}_{ns}} $ represent vectors keeping significant and insignificant terms, respectively.

First, we introduce the definition of Restricted Isometric Property (RIP) for performance analysis as:
\begin{definition}
\label{definition:RIP}
Let $\Phi \in \mathbb{C}^{r \times d}$ and $a_{m} \leq d$. Suppose there exists a restricted isometry constant (RIC) $\delta_{a_{m}}$ of a matrix $\Phi$ such that for each $r \times a$ submatrix $\Phi_{a_{m}}$ of $\Phi$ and for every $\bm{s}$ we have:
 $$ (1-\delta_{a_{m}})\| \bm{s} \|^{2} \leq \| \bm{\Phi_{a_{m}} s} \|^{2} \leq (1+\delta_{a_{m}})\| \bm{s} \|^{2}. $$
The matrix $\Phi$ is said to satisfy the $a_{m}$-restricted isometry property.
\end{definition}
In addition, the performance analysis of SP \cite{Dai2009} is shown in Theorem \ref{theorem:sp performance}.
\begin{theorem}
\label{theorem:sp performance}
Let $\bm{s} \in \mathbb{C}^{d}$ be generally $a_{m}$-sparse and let $\bm{y}=\Phi (\bm{s}+\bm{\eta})+\bm{e}$.
Suppose that the sampling matrix satisfies RIP with parameter $\delta_{6a_{m}}<0.083$.
Then,
\small
\begin{equation}
   \| \bm{s} - \bm{s}_{out} \|_2 \leq \frac{1 + \delta_{6a_{m}}}{\delta_{6a_{m}} (1 - \delta_{6a_{m}})}\left( \| \bm{e} \|_2 +\sqrt{\frac{1+\delta_{6a_{m}}}{a_{m}}}\| \bm{s} - \bm{s}_{a_{m}} \|_1 \right),
\label{Eq: SP error}
\end{equation}
\normalsize
where $\bm{s}_{out}$ is the output of SP and $\bm{s}_{a_{m}}$ is an $a_{m}$-sparse vector minimizing $\| \bm{s} - \bm{s}_{a_{m}} \|_{2}$.
\end{theorem}
It should be noted that, in our case, we formulate the aliasing problem at each downsampled frequency as CS problem solved by SP.
Thus, Theorem \ref{theorem:sp performance} is applied to analyze reconstruction error at each downsampled frequency.
By summing the errors with respect to all downsampled frequencies, we show the total error is still bounded.
\begin{theorem}
\label{theorem:sFFT performance without no pruning}
Let $\bm{\hat{x}} \in \mathbb{C}^{N}$ be generally $K$-sparse.
Given $a_{m}$, $r=O(a_{m})$, $d=\frac{N}{\mu K}$, and $\bm{\hat{x}}_{out}$ is the output of sFFT-DT.
If $\bm{\Phi}$ in Eq. (\ref{eq:CS Formulation}) satisfies RIP with parameter $\delta_{6a_{m}}<0.083$, then we obtain recovery error
$$\| \bm{\hat{x}} - \bm{\hat{x}}_{out} \|_2 \leq C_{\delta_{6a_{m}}}\sqrt{ \| \bm{\hat{x}}_{ns}\|_2^2 + \tau(   \frac{N}{\mu} -Ka_{m})\|  \bm{\hat{x}}  \|_\infty^{2} } ,$$ where
$$C_{\delta_{6a_{m}}} = \frac{1 + \delta_{6a_{m}}}{\delta_{6a_{m}}
(1 -\delta_{6a_{m}})}\sqrt{\frac{d(1+\delta_{6a_{m})}}{a_{m}}},$$
with probability at least $1-\left(\frac{d^{a_{m}}K^{a_{m}}  e^{a_{m}+2}}{\tau N^{a_{m}} (a_{m}+1)^{a_{m}+1}}\right)^{\tau K} $ and computational complexity $\max \left( O(a_{m}\frac{N}{d} \log \frac{N}{d} ), O(a_{m}^2 N) \right)$ by solving the aliasing problem in sFFT-DT with SP.
\end{theorem}
\begin{proof}
First, we relax the term $\| \bm{s} - \bm{s}_{a_{m}} \|_1$ in Eq. (\ref{Eq: SP error}) as $\sqrt{d}\| \bm{s} - \bm{s}_{a_{m}} \|_2$.
Let  $\bm{s}^{i}$, $\bm{\eta}^{i}$, and $a^{i}$ represent the significant terms, insignificant terms, and number of collisions at $i$'th downsampled frequency, respectively.
Thus, $\bm{s}^{i}- \bm{s}^{i}_{out}$ is the reconstruction error at $i$'th downsampled frequency.
It should be noted that, in our case, $\bm{e}=\bm{0}$
%shown in Eq. (\ref{Eq: SP error})
for all downsampled frequencies.
The total error is derived as:
\begin{equation}
\begin{aligned}
\| \bm{\hat{x}} - \bm{\hat{x}}_{out} \|_2^{2}  = \sum_{i=0}^{d-1} \|  \bm{s}^{i}- \bm{s}^{i}_{out}  \|_2^2 \leq  \left(  \frac{1 + \delta_{6a_{m}}}{\delta_{6a_{m}} (1 - \delta_{6a_{m}})}\sqrt{\frac{d(1+\delta_{6a_{m})}}{a_{m}}} \right)^{2}  \sum_{i=0}^{d-1} \|  \bm{s}^{i}- \bm{s}^{i}_{a_{m}}  \|_2^2.
\end{aligned}
\label{eq:derivation of sFFT error}
\end{equation}
For each downsampled frequency labeled $i$, we have (1)  if $a^{i} \leq a_{m}$, then $\| \bm{s}^{i} - \bm{s}_{a_{m}}^{i} \|_2 \leq \| \bm{\eta}^{i} \|_2$ and
(2) if $a^{i} > a_{m}$, then $\bm{s}^{i} - \bm{s}_{a_{m}}^{i} = \bm{s}^{i}+\bm{\eta}^{i}-\mathcal{T}_{a_{m}}(\bm{s}^{i})$,
where $\mathcal{T}_{a_{m}}(\cdot)$ is a soft-thresholding operator keeping the first $a_{m}$ largest entries in magnitude and setting the others to zero.
Based on the above conditions, we can derive
\begin{equation}
\begin{aligned}
  &\sum_{i=0}^{d-1} \|  \bm{s}^{i}- \bm{s}^{i}_{a_{m}}  \|_2^2\\
    &= \sum_{ \{i | a^{i} \leq a_m\} } \| \bm{\eta}^{i} \|_{2}^{2}+   \sum_{ \{i | a^{i} > a_m\} }  \| \bm{s}^{i}+\bm{\eta}^{i}-\mathcal{T}_{a_{m}}(\bm{s}^{i}) \|_{2}^{2} \\
   &\leq \sum_{ \{i | a^{i} \leq a_m\} } \| \bm{\eta}^{i} \|_{2}^{2} + \sum_{ \{i | a^{i} > a_m\} } \| \bm{\eta}^{i} \|_{2}^{2} + \| \bm{s}^{i}-\mathcal{T}_{a_{m}}(\bm{s}^{i}) \|_{2}^{2}\\
   &= \sum_{i=0}^{d-1} \| \bm{\eta}^{i} \|_{2}^{2} +  \sum_{ \{i | a^{i} > a_m\} } \| \bm{s}^{i}-\mathcal{T}_{a_{m}}(\bm{s}^{i}) \|_{2}^{2} \\
   &\leq \| \bm{\hat{x}}_{ns} \|_{2}^{2} + \tau(\frac{N}{\mu}-Ka_{m})\| \bm{\hat{x}} \|_{\infty}^{2}.
\end{aligned}
\label{Eq: Derive the error based on am}
\end{equation}
The last inequality is due to the fact that $\bm{\eta}^{i}$ is the vector keeping the insignificant terms such that $ \sum_{i=0}^{d-1} \|  \bm{\eta}^{i} \|_2^2 = \| \bm{\hat{x}}_{ns}\|$.
In addition, for $\bm{s}^{i}$ with $a^{i} > a_m$, it contains at most $d$ significant frequencies in the worst case.
Thus, $\bm{s}^{i}-\mathcal{T}_{a_{m}}(\bm{s}^{i})$ leaves $d-a_{m}$ significant frequencies and
$\| \bm{s}^{i}-\mathcal{T}_{a_{m}}(\bm{s}^{i}) \|_2^2 \leq  (d-a_{m})\| \bm{s}^{i} \|_{\infty}^2 \leq (d-a_{m})\| \bm{\hat{x}} \|_{\infty}^2 $ always holds.
% for all $i \in \{i | a^{i} > a_m\}$.

On the other hand, when Theorem \ref{theorem:probability of sfftdt v2} holds, it implies at most $\tau K -1$ downsampled frequencies with number of collisions larger than $a_m$.
Then, the cardinality of $\{i | a^{i} > a_m\}$ is $\tau K -1$ and $\sum_{ \{i | a^{i} > a_m\} } \| \mathcal{T}_{a_{m}}(\bm{s}) \|_2^2 \leq   (\tau K -1)(d-a_{m})\| \bm{\hat{x}} \|_\infty^2 \leq   \tau K(d-a_{m})\| \bm{\hat{x}} \|_\infty^2 \leq   \tau(\frac{N}{\mu}-Ka_{m})\| \bm{\hat{x}} \|_\infty^2$.

Finally, the computational cost consists of the costs of generating the required syndromes and running SP.
In similar to exactly-$K$ sparse case, given $a_{m}$ and $d=O(\frac{N}{K})$, FFTs for downsampled signals totally cost $O(a_{m} \frac{N}{d} \log \frac{N}{d})$ and the computational cost of SP discussed in Sec. \ref{Sec: Refinement} is $O(a_{m}^2 N)$.
Thus, we complete this proof.
\end{proof}

It is important to check if the sufficient condition $\delta_{6a}<0.083$ holds.
How to compute RIC for a matrix, however, is a NP-hard problem.
But we can know that RIC is actually related to $a_{m}$, $r$, and $d$.
When fixing $r = 3a_{m}$, $d$ must satisfy $O(a_{m})$ such that $\delta_{6a}<0.083$ holds with high probability, implying that sFFT-DT works well under $K=\Theta(N)$.
The term, $\tau(   \frac{N}{\mu} -Ka_{m})\|  \bm{\hat{x}}  \|_\infty^{2}$, in Theorem \ref{theorem:sFFT performance without no pruning} also reveals that $K=\Theta(N)$ results in better performance.
To further improve this result in Theorem \ref{theorem:sFFT performance without no pruning}, we propose a pruning strategy to prune $\Phi \in \mathbb{C}^{r \times d}$ into $\Phi_{p} \in \mathbb{C}^{r \times O(a)}$, which benefits the recovery performance and computational cost.
Specifically, we are able to reduce recovery errors and easily achieve the sufficient condition if $d$ is replaced by $O(a)$ and reduce the computational cost of SP to become $O(K)$.

\subsection{Pruning Strategy}\label{Sec: decide a}
One can observe from Eq. (\ref{eq:CS Formulation}) that the $j$'th column of $\Phi$ corresponds to the location $k+(j-1)\frac{N}{d}$ and the root $e^{\frac{-i2\pi (k+(j-1)\frac{N}{d})l}{N}}$.
The basic idea of pruning is to prune as many locations/roots corresponding to insignificant terms as possible.
Here, the pruning strategy contains three steps:
\begin{itemize}
  \item[1.] Estimate the number of collisions, $a$, for each downsampled frequency by singular value decomposition (SVD) for the matrix $\bm{M}$ presented in Eq. (\ref{eq:auxiliary fun of MPT}).
  \item[2.] Form the polynomial presented in Step (ii) of syndrome decoding, and substitute all roots in $U_{k}$ into the polynomial and reserve the locations with the first $O(a)$ smallest errors in $U_{k}$.
  \item[3.] According to reserved locations, prune $\Phi$ to yield $\Phi_{p}$.
\end{itemize}

\subsubsection{Step 1 of Pruning}
In Step 1 of the pruning strategy, we estimate  $a$'s for all downsampled frequencies. By doing SVD for the matrix in Eq. (\ref{eq:auxiliary fun of MPT}), there are $a_{m}$ singular values for each downsampled frequency. Collect all $\frac{N}{d}a_{m}$ singular values from all downsampled frequencies and index each singular value according to which downsampled frequency it is from.
The first $K$ largest singular values will vote which downsampled frequency includes the significant term.

Now, we show why the strategy is effective.
We redefine the problem in Eq. (\ref{eq:auxiliary fun of MPT}) for generally $K$-sparse signals.
Let $S_{s}$ be a set containing all indices of significant terms and let $S_{ns}$ be the one defined for insignificant terms, where $S_{s} \cap S_{ns} = \{ \}$, $S_{s} \cup  S_{ns} = \{ 0,1, \ ... \ N -1 \}$, $|S_{s}| = a$ and $|S_{ns}| = d-a$.
Thus, the syndrome in Eq. (\ref{eq:MPT obejective fun}) can be rewritten as $ m_{i} = p_{0}z_{0}^{i}+p_{1}z_{1}^{i}+...+p_{N-1}z_{N-1}^{i} = m_{s,i}+m_{ns,i}$, where
$  m_{s,i} = \sum_{j \in S_{s} } p_{j}z_{j}^{i}$ and $ m_{ns,i} = \sum_{k \in S_{ns} } p_{k}z_{k}^{i}$.
Now, the matrix $\bm{M}$ in Eq. (\ref{eq:auxiliary fun of MPT}) can be rewritten as:
\small
\begin{equation}
\begin{aligned}
\bm{M}=
\begin{bmatrix}
       m_{0}  &  \cdots   &   m_{a_{m}-1}    \\
       m_{1}          & \cdots  &  m_{a_{m}} \\
        \vdots            & \ddots  &\\
       m_{a_{m}-1}            & \cdots   & m_{2a_{m}-2}
\end{bmatrix}=\bm{M}_{s}+\bm{M}_{ns}=
\begin{bmatrix}
       m_{s,0}+m_{ns,0}  &  \cdots   &   m_{s,a_{m}-1}+m_{ns,a_{m}-1}    \\
       m_{s,1}+m_{ns,1}          & \cdots  &  m_{s,a_{m}}+m_{ns,a_{m}} \\
        \vdots            & \ddots  &\\
       m_{s,a_{m}-1}+m_{ns,a_{m}-1}            & \cdots   & m_{s,2a_{m}-2}+m_{ns,2a_{m}-2}
\end{bmatrix},
\end{aligned}
\label{eq:formulation of M for general }
\end{equation}
\normalsize
where $\bm{M}_{ns}$ acts like a ``noise'' matrix produced by insignificant components and $\bm{M}_{s}$ comes from significant terms.
%Thus, it is apparent that  $  \frac{\|\bm{M}_{s}\|_{F}}{\|\bm{M}_{ns}\|_{F}}$, where $\| \cdot \|_{F}$ denotes the Frobenius norm, is proportional to $SNR_{ori}$, as defined in Eq. (\ref{Eq: SNR_ori}).
%In addition, letting $n_{j}=n_{0}+vj$ for $j \in [0,2a_{m}-2]$, where $v$ is a constant (discussed later),
$\bm{M}$ can also be expressed as:
\small
\begin{equation}
\begin{aligned}
\bm{M}=\sum_{i \in S_{s}\cup S_{ns}}p_{i}\bm{z}_{i}\bm{z}_{i}^{T}= \bm{M}_{s}+\bm{M}_{ns}=\sum_{j \in S_s} p_{j}\bm{z}_{j}\bm{z}_{j}^{T}
+\sum_{k \in S_{ns}} p_{k}\bm{z}_{k}\bm{z}_{k}^{T},
\label{eq:formulation of another M for general}
\end{aligned}
\end{equation}
\normalsize
where $  \bm{z}_{i} = [ z_{i}^{0} \ z_{i}^{1} \ ... \ z_{i}^{a_{m}-1}   ]^{T}$  is a column vector, as defined in Eq. (\ref{eq:MPT obejective fun}).

It is worth noting that Eq. (\ref{eq:formulation of another M for general}) is similar to SVD.
Nevertheless, there are some differences between them:
(1) $p_{i}$'s are complex but not real and $\|\bm{z}_{i}\|_{2}$'s are not normalized;
(2) $\bm{z}_{i}$'s are not orthogonal vectors; and
(3) The actual SVD of $\bm{M}$ is $\bm{M}=\sum_{i \in S_{s}\cup S_{ns}}p_{i}\bm{z}_{i}\bm{z}_{i}^{*}$, where $*$ denotes a conjugate transpose.

To alleviate the difference (1), Eq. (\ref{eq:formulation of another M for general}) is rewritten as:
\small
\begin{equation}
\bm{M}=\sum_{j \in S_s} a_{m}|\hat{p}_{j}|\hat{\bm{z}}_{j}\hat{\bm{z}}_{j}^{T}+\sum_{k \in S_{ns}} a_{m}|\hat{p}_{k}|\hat{\bm{z}}_{k}\hat{\bm{z}}_{k}^{T},
\label{eq:transformed formulation of another M for general}
\end{equation}
\normalsize
where $  \hat{\bm{z}}_{i} = \frac{1}{\sqrt{a_{m}}}e^{\sqrt{-1}\theta_{i}}[ z_{i}^{0} \ z_{i}^{1} \ ... \ z_{i}^{a_{m}-1}   ]^{T}$ and $\hat{p}_{i}= |p_{i}|e^{-\sqrt{-1}\theta_{i}}$ for $i \in S_{s}\cup S_{ns}$.
Thus, we have $\| \hat{\bm{z}}_{i} \|_{2} =1$.
As for the difference (3), it can be solved by symmetric SVD (SSVD) \cite{Angelika1988} instead of SVD.
However, the singular values of SSVD have been proven to be the same as those of SVD for the same matrix.
Thus, Eq. (\ref{eq:transformed formulation of another M for general}) can directly use SVD for matrix $\bm{M}$ to obtain singular values.

On the other hand, the difference (2) is inevitable since $\hat{\bm{z}}_{i}$'s are not orthogonal, leading to the fact that the singular values of $\bm{M}$ are not directly equal to the magnitudes of frequencies $|p_{i}|$'s.
However, they are actually related.
%both singular values and \textcolor{red}{magnitudes of frequencies} are still \textcolor{red}{related}.
In \cite{Takos2008}, Takos and Hadjicostis actually explore the relationship between the eigenvalues of $\bm{M}$ and signal values.
But their proofs are based on real BCH codes, implying $\bm{M}$ is a Hermitian matrix.
It is not appropriate for our case because $\bm{M}$ is, in fact, a complex symmetric matrix.
Thus, we develop another theorem illustrating the relationship between singular value and signal values.
First, Lemma \ref{lemma:singular value of sum of two matrices} in \cite{Hogben2007} illustrates the singular values of sum of matrices.
\begin{lemma}
\label{lemma:singular value of sum of two matrices}
For any matrices $\bm{A} \in \mathbb{C}^{n \times n}$ and $\bm{B}\in \mathbb{C}^{n \times n}$, let $\bm{C}=\bm{A}+\bm{B}$ and let $\sigma_{j}(\cdot)$ be a function returning the $j$'th largest singular value with $\sigma_{1}(\bm{C})\geq \sigma_{2}(\bm{C}) \geq ... \geq \sigma_{n}(\bm{C})$.
Then,
\begin{equation}
\begin{aligned}
\sigma_{j}(\bm{A})-\sigma_{1}(\bm{B}) \leq \sigma_{j}(\bm{C})
 \leq \sigma_{j}(\bm{A})+\sigma_{1}(\bm{B}),
\label{eq:singular value formula}
\end{aligned}
\end{equation}
holds for $1\leq j\leq n$.
\end{lemma}
Second, for both matrices $\bm{M}_s$ and $\bm{M}_{ns}$ given in Eq. (\ref{eq:formulation of another M for general}), we explore the upper bound of singular values of $\bm{M}_{ns}$ and the lower bound of singular values of $\bm{M}_{s}$, where both bounds are used as the sufficient condition of correctly determining the number $a$ of collisions in aliasing.
Specifically, Lemma \ref{lemma:singular value of sum of two matrices} is used to derive the upper bound of singular values of $\bm{M}_{ns}$ whatever $a$ is.
But the lower bound of singular values of $\bm{M}_{s}$ is non-trivial only when $a=1$ because it becomes 0 for $a>1$.
\begin{lemma}
\label{lemma:singular value and M}
For any $a$, the singular values of $\bm{M}_{ns}$ satisfy $$ \sigma_{1}(\bm{M}_{ns}) \leq a_{m}dq_{max},$$
where $q_{max} = \max_{k \in S_{ns}} |\bm{\hat{x}}[k]|$.
In addition, for $a=1$, the singular values of $\bm{M}_{s}$ satisfy
$$ \sigma_{1}(\bm{M_{s}}) \geq a_{m}\epsilon_{min},$$
where $\epsilon_{min} = \min_{j \in S_{s}} |\bm{\hat{x}}[j]|$.
\end{lemma}
\begin{proof}
Since $\bm{M}_{ns}=\sum_{k \in S_{ns}  }  a_{m}|p_{k}|\hat{\bm{z}}_{k}\hat{\bm{z}}_{k}^{T}$, we have
\begin{equation}
\begin{aligned}
&\sigma_{1}(\bm{M}_{ns}) = \sigma_{1}( \sum_{k \in S_{ns}  }  a_{m}|p_{k}|\hat{\bm{z}}_{k}\hat{\bm{z}}_{k}^{T})\\
&\leq \sum_{k \in S_{ns}  } \sigma_{1}(  a_{m}|p_{k}|\hat{\bm{z}}_{k}\hat{\bm{z}}_{k}^{T}) = \sum_{k \in S_{ns}  } a_{m}|p_{k}| \\
&\leq a_{m}(d-a)q_{max} \leq a_{m}dq_{max}.
\label{eq:max singular value of insignificant term}
\end{aligned}
\end{equation}
Similarity, for $a=1$, $\bm{M}_{s}= a_{m}|p_{j}|\hat{\bm{z}}_{j}\hat{\bm{z}}_{j}^{T}$ with $j \in S_{s}$. Then, $\sigma_{1}(\bm{M}_{s})=a_{m}|p_{j}| \geq a_{m}\epsilon_{min}$. We complete this proof.
\end{proof}

Combined with Lemma \ref{lemma:singular value of sum of two matrices} and Lemma \ref{lemma:singular value and M}, we can derive the following theorem.
\begin{theorem}
\label{theorem:successfuly determine a}
If $ \epsilon_{min} > 2dq_{max}$ and all downsampled frequencies satisfy $a \leq 1$, then sFFT-DT correctly decides the number of collisions.
\end{theorem}
\begin{proof}
For all downsampled frequencies satisfying $a = 1$, Lemma \ref{lemma:singular value of sum of two matrices} and Lemma \ref{lemma:singular value and M} induce the fact:
$$a_{m}\epsilon_{min}- a_{m}dq_{max} \leq \sigma_{1}(\bm{M}_{s})-\sigma_{1}(\bm{M}_{ns}) \leq \sigma_{1}(\bm{M}).$$
In addition, for all downsampled frequencies with $a=0$,
$$ \sigma_{1}(\bm{M})= \sigma_{1}(\bm{M}_{ns}) \leq a_{m}dq_{max}.$$
As a result, if $a_{m}dq_{max} < a_{m}\epsilon_{min}- a_{m}dq_{max}$, it implies that sFFT-DT can correctly determine the number of collisions, ({\em i.e.}, distinguish the downsampled frequencies with $a=0$ and those with $a=1$, by finding the first $K$ largest singular values).
\end{proof}

Remark: It should be noted that Theorem \ref{theorem:successfuly determine a} only holds for all downsampled frequencies with $a \leq 1$. The probability that there is no downsampled frequency with $a > 1$ shown in Lemma \ref{lemma:probability of bin} is at most $Pr(d,1)$. By setting $d$ larger, it means that Theorem \ref{theorem:successfuly determine a} holds with higher probability as $1-Pr(d,1)$.

In fact, Theorem \ref{theorem:successfuly determine a} also reveals that fact that it is more difficult to satisfy $\epsilon_{min} > 2dq_{max} $ when $d$ becomes large enough.
In other words, one needs to force $K=\Theta(N)$ such that $d=O(\frac{N}{K})=O(1)$ in order to correctly determine $a$.

\subsubsection{Step 2 of Pruning}
After determining $a$'s for all downsampled frequencies, Step 2 of the pruning strategy runs the following procedure to know which locations should be pruned:
\begin{itemize}
  \item[(a).] Solve $ \tilde{c} = (\bm{M}_{s}+\bm{M}_{ns})^{-1}(\bm{m}_{s}+\bm{m}_{s})$.
  \item[(b).] Let $\tilde{P}(z)=z^{a}+\tilde{c}[a-1]\tilde{z}^{a-1}+...+\tilde{c}[1]z+\tilde{c}[0].$
  \item[(c).] $\tilde{Z}$ is the set of collecting all $z \in U_{k}$ with the first $O(a)$ smallest $|\tilde{P}(z)|$.
\end{itemize}
This procedure is similar to syndrome decoding except that Step (ii) in Sec. \ref{ssec:MPT} is changed.
As mentioned above, due to the ill-conditioned problem such as Wilkinson's polynomial, the problem of approximating the roots, given the coefficients with noisy perturbation, is ill-conditioned.
Instead of finding roots by solving the polynomial, since the set including all candidate roots, $U_{k}$, is finite, we substitute all candidate roots in $U_{k}$ into the polynomial and store the roots with the first $O(a)$ smallest errors in the set $\tilde{Z}$, as also adopted in \cite{Takos2008}.

\subsubsection{Step 3 of Pruning}
By feeding $\tilde{Z}$ into Step 3 of the pruning strategy, we can decide which columns of $\bm{\Phi}$ should be pruned according to the following criterion.
If the root belongs to $\tilde{Z}$, its corresponding column is preserved; otherwise, it is pruned.
Finally, let $\Phi_{p}$ be the outcome after pruning and let it be used to replace $\Phi$ in Eq. (\ref{eq:CS Formulation}).

\subsection{Non-iterative sFFT-DT: Algorithm for Generally K-Sparse Signals}\label{Sec: algorithm for general K sparse}

For generally $K$-sparse signals, sFFT-DT solves the aliasing problem once, as shown in Algorithm 2, which integrates the pruning strategy and CS-based approach.
The function \textbf{main} contains four parts.
For clarity, Fig. \ref{fig:flowchart} illustrates the flowchart of sFFT-DT for generally $K$-sparse signals and we describe each part as follows.
Part 1: In Lines 2-9, several downsampled signals are generated for performing FFTs with different shift factors.
Specifically, the downsampled signals in Lines 2-5 are prepared for the pruning strategy and those in Lines 6-9 are used for CS recovery problem.
To distinguish between these two, the signals for Lines 2-5 are represented by $x_{d}$ and those for Lines 6-9 are represented by $x_{s}$.
Part 2: Lines 11-21 run the Step 1 of the pruning strategy and decide the number of significant terms in the downsampled frequencies. $V_{sin}$ is a set used to save all singular values of $\bm{M}$'s (defined in Eq. (\ref{eq:formulation of M for general })) corresponding to frequencies.
Part 3: Lines 23-25 run the Step 2 and Step 3 of the pruning strategy, where $\tilde{Z}$ collects the roots corresponding to insignificant terms.
According to $\tilde{Z}$, we can prune $\Phi$ and output $\Phi_{p}$.
Part 4: Given $\Phi_{p}$, Lines 26-28 solve the CS recovery problem by Subspace Pursuit, as mentioned in Sec. \ref{Sec: Refinement}.

\begin{figure}[t]
\begin{minipage}[b]{.98\linewidth}
%  \centering
\centering{\epsfig{figure=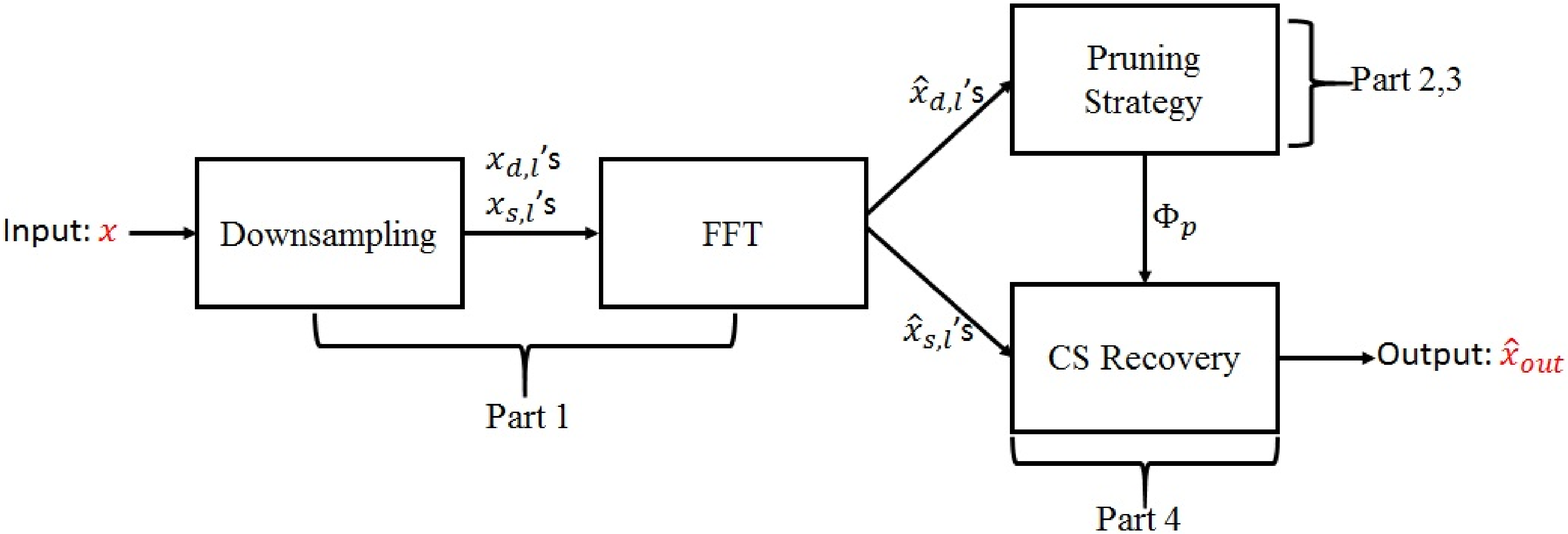,width=5.0 in}}
\end{minipage}
\hfill
\caption{Flowchart of sFFT-DT for generally $K$-sparse signals.}
\label{fig:flowchart}
\end{figure}

\begin{algorithm}[h]
\fontsize{11pt}{0.9em}\selectfont
\setlength{\abovecaptionskip}{0pt}
\setlength{\belowcaptionskip}{0pt}
\caption{sFFT-DT for generally $K$-sparse signals.}
\label{Table:our algorithm for general K-sparse signal}
\begin{tabular}[t]{p{17.7cm}l}
%\hline
\textbf{Input:} $\bm{x}$, $K$;\quad \textbf{Output:} $\bm{\hat{x}}_{out}$; \\
\textbf{Initialization:} $\bm{\hat{x}}_{out}=\mathbf{0}$, $  d=O(\frac{N}{K})$, $R=\{ \}$, $V_{sin}=\{ \}$, $a_{m}$; \\
\hline\hline
01. \textbf{function} \textbf{main}()\\
02. \quad\textbf{for} $l = 0$ to $a_{m} -1$ \\
03. \quad\quad $x_{d,2l}[k]=x[dk+2l]$ for $  k\in [0,\frac{N}{d}-1]$;\\
04. \quad\quad $x_{d,2l+1}[k]=x[dk+2l+1]$ for $  k\in [0,\frac{N}{d}-1]$;\\
05. \quad\textbf{end for}\\
06. \quad Generate $\{n_{0},n_{1},...n_{3a_{m}-1} \}$ in Sec. \ref{Sec: Refinement};\\
07. \quad\textbf{for} $l = 0$ to $3a_{m} -1$ \\
08. \quad\quad $x_{s,l}[k]=x[dk+n_{l}]$ for $  k\in [0,\frac{N}{d}-1]$;\\
09. \quad\textbf{end for}\\
10. \quad Do FFT of all $\bm{x}_{d}$'s, $\bm{x}_{s}$'s to obtain $\bm{\hat{x}}_{d}$'s and $\bm{\hat{x}}_{s}$'s.\\
11. \quad \textbf{for} $k = 0$ to $  \frac{N}{d}-1$\\
12. \quad\quad \quad $m_{j}= \hat{x}_{d,j}[k]$ for $j \in [0,2a_{m}-1]$; \\
13. \quad\quad \quad Use $m_{j}$'s to form $\bm{M}$ defined in Sec. \ref{Sec: decide a}; \\
14. \quad\quad \quad Do SVD of $\bm{M}$ and put singular values into \\
\quad\quad\quad\quad \quad  the set $V_{sin}$;    \\
15. \quad \textbf{end for} \\
16. \quad Find the first $K$ largest singular values from \\
\quad\quad\quad $V_{sin}$ and save them as $\sigma_{1}, \sigma_{2},...,\sigma_{K}$.\\
17. \quad \textbf{for} $l = 1$ to $K$\\
18. \quad\quad \textbf{if} ($\sigma_{l}$ originates from the $k$'th frequency) \\
19. \quad \quad \quad $a$ of the $k$'th frequency increases by 1;\\
20. \quad \quad \textbf{end if}\\
21. \quad \textbf{end for} \\
22. \quad \textbf{for} $k = 0$ to $  \frac{N}{d}-1$\\
23. \quad \quad $m_{j}= \hat{x}_{d,j}[k]$ for $j \in [0,2a_{m}-1]$; \\
24. \quad \quad Run Step 2 and 3 of pruning strategy in Sec. \ref{Sec: decide a};\\
25. \quad \quad and output $\Phi_{p}$.\\
26. \quad \quad $m_{j}= \hat{x}_{s,j}[k]$ for $j \in [0,3a_{m}-1]$; \\
27. \quad \quad Solve Eq. (\ref{eq:CS Formulation}) given $\Phi_{p}$ by SP and assign  \\
28. \quad \quad $\bm{\hat{x}}_{out}[k+j\frac{N}{d}]=$$\bm{s}_{out}[j]$ for $j \in [0,d-1]$. \\
29. \quad \textbf{end for} \\
30. \textbf{end} \textbf{function}\\
%40. \textbf{end} \textbf{function}\\
\hline
\end{tabular}
\end{algorithm}
\renewcommand\arraystretch{1}

\begin{table}[t]
\fontsize{7.5pt}{1em}\selectfont
\centering
\setlength{\abovecaptionskip}{0pt}
\setlength{\belowcaptionskip}{4pt}
\caption{The effect of pruning in terms of computational cost and recovery performance under $N=2^{24}$ and $SNR(\bm{\hat{x}}_{s})=30$ $dB$.}
\label{Table: pruning comparision}
\doublerulesep=2pt
\begin{tabular}[tc]{|c|c|c|c|c|c|c|c|c|c|c|c|c|c|}
\hline
{\LARGE \textcolor{white}{o}}K& $2^{8}$& $2^{9}$& $2^{10}$ & $2^{11}$& $2^{12}$ & $2^{13}$ & $2^{14}$ & $2^{15}$ & $2^{16}$ & $2^{17}$ & $2^{18}$ & $2^{19}$ & $2^{20}$ \\ \hline
Time Cost without Pruning (Sec)& 5.731& 4.287& 3.315 & 3.813 & 4.459 & 8.681 & 15.10 & 23.61 & 50.27 & 101.22 & 217.28 & 463.21 & 989.41  \\ \hline
Time Cost with Pruning (Sec)& 0.021&  0.022 & 0.033 & 0.053 & 0.101 & 0.211 & 0.321 & 0.674 & 1.237 & 2.524 & 5.138 & 9.918 &19.539 \\ \hline
$SNR(\bm{\hat{x}}_{out})$ without Pruning (dB)& -66.1& -51.9& -36.4 & -24.6 & -13.9 & -2.37 & 11.34 & 21.6 & 28.7 & 29.3 & 29.7 & 29.9 & 29.9 \\ \hline
$SNR(\bm{\hat{x}}_{out})$  with Pruning (dB)& 4.67& 10.1& 14.8 & 20.1 & 23.1 & 24.9 & 27.7 & 29.7 & 29.9 & 29.9 & 29.9 & 29.9 & 29.9 \\ \hline
\end{tabular}
\end{table}

\begin{table}[t]
\fontsize{7.5pt}{1em}\selectfont
\centering
\setlength{\abovecaptionskip}{0pt}
\setlength{\belowcaptionskip}{4pt}
\caption{The effect of pruning in terms of computational cost and recovery performance under $N=2^{24}$ and $SNR(\bm{\hat{x}}_{s})=20$ $dB$.}
\label{Table: pruning comparision 2}
\doublerulesep=2pt
\begin{tabular}[tc]{|c|c|c|c|c|c|c|c|c|c|c|c|c|c|}
\hline
{\LARGE \textcolor{white}{o}}K& $2^{8}$& $2^{9}$& $2^{10}$ & $2^{11}$& $2^{12}$ & $2^{13}$ & $2^{14}$ & $2^{15}$ & $2^{16}$ & $2^{17}$ & $2^{18}$ & $2^{19}$ & $2^{20}$\\ \hline
Time Cost without Pruning (Sec)& 5.693& 4.436& 3.761 & 3.903 & 4.634 & 8.511 & 16.20 & 31.61 & 51.92 & 108.49 & 229.31 & 492.01 &1032.94 \\ \hline
Time Cost with Pruning (Sec)& 0.021&  0.023 & 0.031 & 0.056 & 0.097 & 0.187 & 0.335 & 0.622 & 1.343 & 2.724 & 5.605 & 10.492 &20.034 \\ \hline
$SNR(\bm{\hat{x}}_{out})$ without Pruning (dB)& -66.4& -53.3& -40.4 & -28.2 & -15.9 & -4.97 & 9.19 & 19.3 & 19.9& 19.9& 19.9& 19.9& 19.9 \\ \hline
$SNR(\bm{\hat{x}}_{out})$  with Pruning (dB)& 0.04& 1.56 & 6.78 & 12.1 & 16.4 & 18.1 & 19.3 & 19.7 & 19.9& 19.9& 19.9& 19.9& 19.9 \\ \hline
\end{tabular}
\end{table}

\begin{table}[t]
\fontsize{7.5pt}{1em}\selectfont
\centering
\setlength{\abovecaptionskip}{0pt}
\setlength{\belowcaptionskip}{4pt}
\caption{The effect of pruning in terms of computational cost and recovery performance under $N=2^{24}$ and $SNR(\bm{\hat{x}}_{s})=10$ $dB$.}
\label{Table: pruning comparision 3}
\doublerulesep=2pt
\begin{tabular}[tc]{|c|c|c|c|c|c|c|c|c|c|c|c|c|c|}
\hline
{\LARGE \textcolor{white}{o}}K& $2^{8}$& $2^{9}$& $2^{10}$ & $2^{11}$& $2^{12}$ & $2^{13}$ & $2^{14}$ & $2^{15}$ & $2^{16}$ & $2^{17}$ & $2^{18}$ & $2^{19}$ & $2^{20}$ \\ \hline
Time Cost without Pruning (Sec)& 5.611& 4.627& 3.802 & 3.892 & 4.561 & 8.639 & 16.39 & 30.32 & 59.14& 124.12 & 273.21 & 522.52& 1095.42 \\ \hline
Time Cost with Pruning (Sec)& 0.023&  0.029 & 0.038 & 0.052 & 0.125 & 0.212 & 0.326 & 0.644 & 1.227 & 2.321 & 4.732 & 9.327 & 19.394\\ \hline
$SNR(\bm{\hat{x}}_{out})$ without Pruning (dB)& -73.1& -60.6& -48.5 & -36.3 & -24.3 & -11.6 & 2.27 & 9.53 & 9.97 & 9.98 & 9.99 & 9.99 & 9.99\\ \hline
$SNR(\bm{\hat{x}}_{out})$  with Pruning (dB)& -1.19& -0.41& 0.85 & 2.59 & 6.03 & 8.83 & 9.69 & 9.94 & 9.98 & 9.99 & 9.99 & 9.99 & 9.99  \\ \hline
\end{tabular}
\end{table}

\begin{figure*}[!t]
\begin{minipage}[b]{.48\linewidth}
  \centering{\epsfig{figure=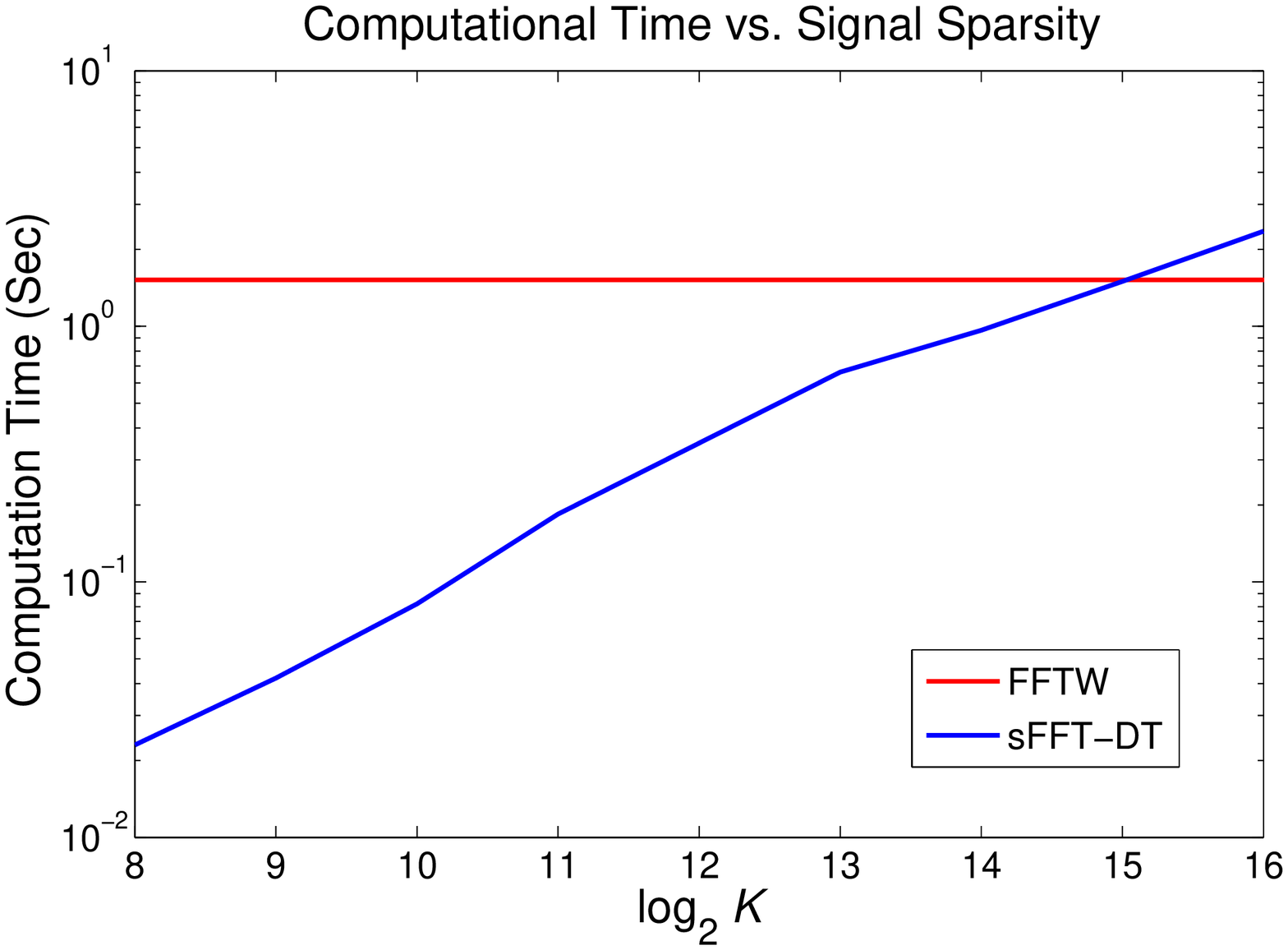,width=3.45in}}
  \centerline{(a)}
\end{minipage}
\begin{minipage}[b]{.48\linewidth}
  \centering{\epsfig{figure=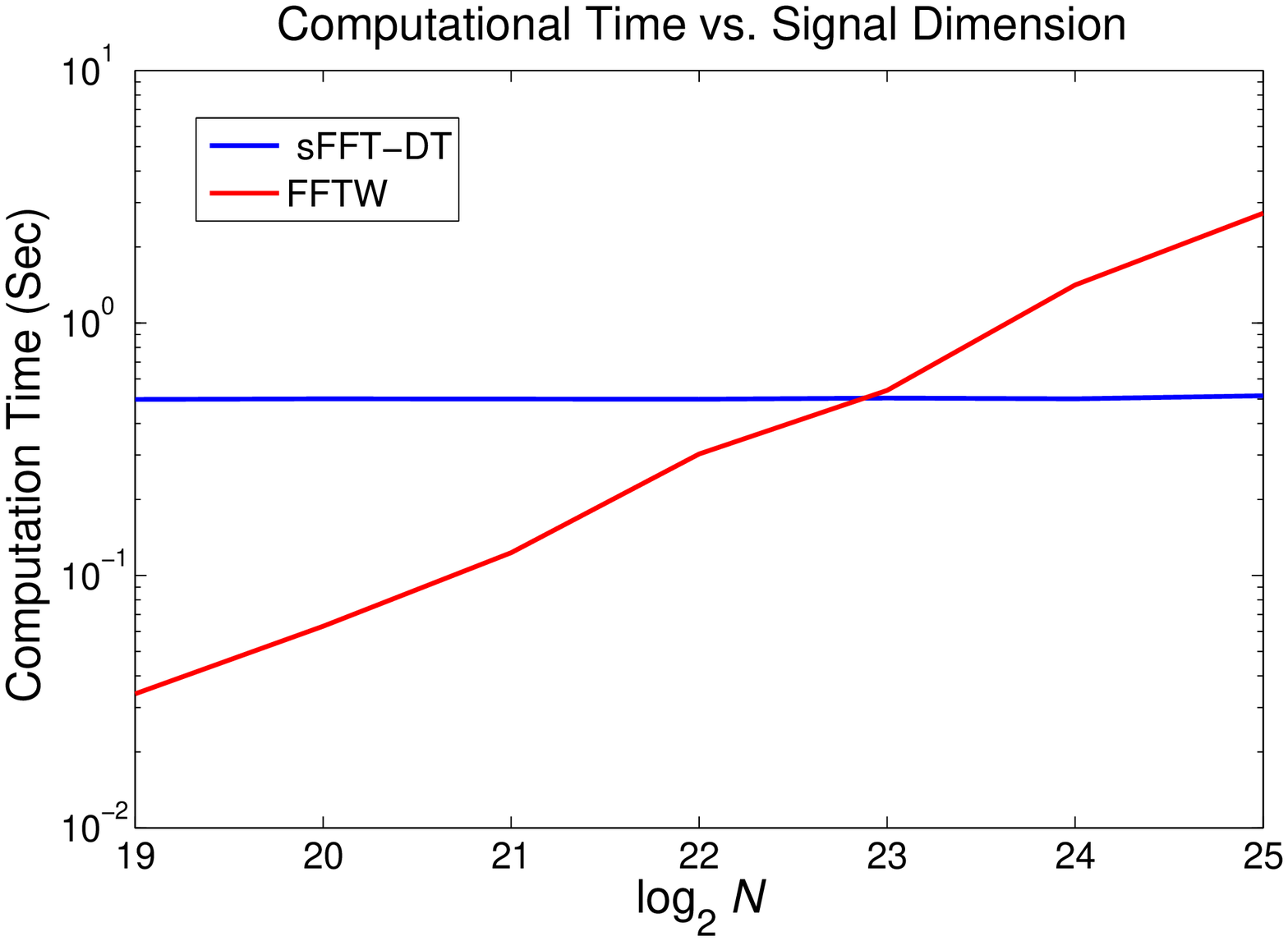,width=3.45in}}
  \centerline{(b)}
\end{minipage}
\hfill
\caption{Comparison between non-iterative sFFT-DT and FFTW for generally $K$-sparse signals. (a) Computational time vs. sparsity under $N=2^{24}$ and $a_{m}=3$. (b) Computational time vs. signal dimension under $K=2^{12}$ and $a_{m}=3$. }
\label{fig:general computational time}
\end{figure*}

\subsection{Computational Complexity of sFFT-DT for Generally $K$-Sparse Signals}\label{Sec: Complexity  for General Ksparse Signal}
In this section, we analyze the computational cost of sFFT-DT for generally K-sparse signals based on Theorem \ref{theorem:probability of sfftdt v2} for the four parts of the Main function. %Four parts defined in Sec. \ref{Sec: algorithm for general K sparse}  are discussed respectively.

Part 1 is to do FFT for downsampled signals, and it costs $  O( a_{m} \frac{N}{d}\log \frac{N}{d})$.
Part 2 solves SVD of $\bm{M} \in \mathbb{C}^{a_{m}\times a_{m}}$ for each downsampled frequency.
Since SVD will totally run $  O(\frac{N}{d})$ times, Part 2 will cost $ O(  \frac{N}{d}  a_{m}^{3})$, according to \cite{Angelika1988}.
Part 3 costs $O(\frac{N}{d} a_{m}^2)$ for computing coefficients of polynomial and $O(N)$ for estimating $|\tilde{P}(z)|$ for all $z \in U_{k}$ in Sec. \ref{Sec: decide a}.
Finally, CS recovery problem in Part 4 depends on the cost of SP.
With the pruning strategy, SP costs $O(a_{m}rd)=O(a_{m}^{3})$.
Thus, the total cost in Part 4 is $O(\frac{N}{d}a_{m}^{3})$ since SP runs $O(\frac{N}{d})$ times, as described in Sec. \ref{Sec: Refinement}.
Thus, the total computational cost of sFFT-DT is bounded by $\max(O( a_{m} \frac{N}{d}\log \frac{N}{d}),O(N))$.

Consequently, the computational cost of sFFT-DT for generally $K$-sparse signals still is impacted by $a_{m}$ and $d$ as in the exactly-$K$ sparse case.
If significant frequencies distribute uniformly, both $a_{m}$ and $d$ can be set based on Theorem \ref{theorem:probability of sfftdt v2}.
In this case, since $a_{m}$ is a constant, the computational cost is bounded by Part 1 and Part 3, which is $\max(O(K \log K),O(N))$.
It should be noted that the Big-O constant of $O(N)$ is very small because only Step 2 of pruning in Line 24 involves $O(N)$ and the operation of estimating $|\tilde{P}(z)|$ for all $z \in U_{k}$ is simple.
Thus, as shown in our experimental results, $O(N)$ does not dominate the computational cost of sFFT-DT.
But the Big-O constants of the generally $K$-sparse case are still larger than those of the exactly $K$-sparse case because the former needs more syndromes.
%
% For example, $a_{m}=3$ (or $4$) together with $\gamma = 4$ r $8$) are enough to solve most aliasing problems with $a < a_{m}$ with very high accuracy level.
%Thus, when little error is tolerable, sFFT-DT costs $O(K \log K)$. Moreover, although $O(a_{m} \gamma K \log \gamma K)$ is better than $K a_{m}^3$ ($a_{m}$ is a constant) from the viewpoint of complexity, $K a_{m}^3$ possibly dominates the whole algorithm in practice.
%Roughly speaking, $K a_{m}^3$ with $a_{m}=3$ will be larger than $K \log K$ when $K <2^{27}$. Thus, in our experiments, it is better to assign smaller $a_{m}$ by adjusting $d$.

\subsection{Simulation Results for Generally $K$-Sparse Signals}\label{Sec: experimental result of General}
The simulation environment is similar to the one described in Sec. \ref{Sec: Experimental Results}.
We only compare sFFT-DT with FFTW because sFFT \cite{Haitham2012}\cite{Haitham2012_1} does not release the code and the code of sFFT for the generally $K$-sparse case is difficult to implement (as mentioned in the footnote on Page 3).
Therefore, no experimental results for generally $K$-sparse signals were shown in their papers or websites.

Here, the test signals were generated from the mixture Gaussian model as:
\small
\begin{equation}
   \hat{x} \sim p\mathcal{N}\left( 0,\sigma_{on}^{2} \right)+\left( 1-p \right) \mathcal{N} \left( 0,\sigma_{off}^{2} \right),
  \label{eq:signalproduce}
\end{equation}
\normalsize
where $  p=\frac{K}{N}$ is the active probability that decides which Gaussian model is used and $\sigma_{on}>\sigma_{off}$. For each test signal, its significant terms is defined as $\bm{\hat{x}}_{s}$, as described in Sec. \ref{Sec: sFFT-DT: General K-Sparse}, and $\bm{\hat{x}}_{out}$ is the output signal obtained from sFFT-DT.
We also define
$SNR(\cdot)$ as:
\small
\begin{equation}
SNR(\bm{\bar{x}})=10 \log_{10}   \frac{MSE(\bm{\bar{x}})}{MSE(\bm{\hat{x}}-\bm{\bar{x}})},
\label{Eq: SNR_ori}
\end{equation}
\normalsize
where $MSE( \cdot )$ is the function of calculating the mean squared error.
If $\bm{\bar{x}}=\bm{\hat{x}}_{s}$, then $SNR(\bm{\hat{x}}_{s})$ means the signal-to-noise ratio between significant terms and insignificant terms.
In our simulations, the parameter setting was $d=\frac{N}{32K}$, $a_{m}=3$, and $SNR(\bm{\hat{x}}_{s})$ ranges from $10$ to $30$ dB.
% The reconstruction quality of sFFT-DT for general $K$-sparse signals depends on $SNR_{ori}$.

Tables \ref{Table: pruning comparision}, \ref{Table: pruning comparision 2}, and \ref{Table: pruning comparision 3}  show the efficiency of pruning.
We can see that sFFT-DT with pruning outperforms its counterpart without pruning in terms of computational cost and recovery performance.
The performance degrades when $\frac{N}{K}$ becomes larger as predicted in Theorem \ref{theorem:sFFT performance without no pruning}.
Moreover, we can observe from Table \ref{Table: pruning comparision} $\sim$ Table \ref{Table: pruning comparision 3} that  no matter $SNR(\bm{\hat{x}}_{out})$ is, the condition for achieving perfect approximation in sFFT-DT, {\em i.e.}, $SNR(\bm{\hat{x}}_{out})\approx SNR(\bm{\hat{x}}_s)$, is always $\frac{N}{K} \leq 2^9$.
The phenomenon is consistent with the reconstruction error bound in Theorem \ref{theorem:sFFT performance without no pruning}. Specifically, the reconstruction error bound, $ C_{\delta_{6a_{m}}}\sqrt{ \| \bm{\hat{x}}_{ns}\|_2^2 + \tau(   \frac{N}{\mu} -Ka_{m})\|  \bm{\hat{x}}  \|_\infty^{2} }$, is affected by $ \|\bm{\hat{x}}_{ns}\|_2^2 $ and  $\tau(   \frac{N}{\mu} -Ka_{m})\|  \bm{\hat{x}}  \|_\infty^{2}$.
However, when $ \frac{N}{K}$ is small, $\tau(   \frac{N}{\mu} -Ka_{m})\|  \bm{\hat{x}}  \|_\infty^{2} \rightarrow 0$ and thus $\| \bm{\hat{x}}-\bm{\hat{x}}_{out} \|_2$ is equal to $C_{\delta_{6a_{m}}} \| \bm{\hat{x}}_{ns}\|_2  $.
In other words, the reconstruction error bound is linear to $\| \bm{\hat{x}}_{ns}\|_2 $.
This is a good property as the reconstruction quality of sFFT-DT is inversely proportional to the energy of insignificant terms, $ \| \bm{\hat{x}}_{ns}\|_2 $.

The comparison of computational time between sFFT-DT and FFTW is depicted in Fig. \ref{fig:general computational time}.
Fig. \ref{fig:general computational time}(a) shows the results of computational time versus signal sparsity under fixed $N$.
It is observed that sFFT-DT is remarkably faster than FFTW, except for the cases with $K \geq 2^{15}$.
Fig. \ref{fig:general computational time}(b) shows the results of computational time versus signal dimension under fixed $K$.
It is apparent that the computational time of sFFT-DT is not related to $N$.

In sum, compared with \cite{Haitham2012}\cite{Haitham2012_1}, the proposed sFFT-DT for generally $K$-sparse signals is the first algorithm with the reasonable Big-O constants and is verified to be faster than FFTW.

\section{Conclusions}\label{Sec: Conclusions}
We have presented new sparse Fast Fourier Transform methods based on downsampling in the time domain (sFFT-DT) for both exactly $K$-sparse and generally $K$-sparse signals in this paper.
The accurate computational cost and theoretical performance lower bound of sFFT-DT are proven for exactly $K$-sparse signals.
We also derive the Big-O constants of computational complexity of sFFT-DT and show that they are smaller than those of MIT's methods \cite{Haitham2012}\cite{Haitham2012_1}\cite{Ghazi2013}.
In addition, sFFT-DT is more hardware-friendly, compared with other algorithms, since all operations of sFFT-DT are linear and involved in an analytical solution.
On the other hand, previous works, such as \cite{Haitham2012}\cite{Haitham2012_1}\cite{Ghazi2013}, are based on the assumption that sparsity $K$ is known in advance.
To address this issue, we proposed a simple solution to estimate $K$ and relax this impractical assumption.
We show that the extra cost for deciding $K$ is the same as that required for sFFT-DT with known $K$.
Moreover, we extend sFFT-DT to generally $K$-sparse signals in this paper.
To solve the interference from insignificant frequencies in aliasing, we first reformulate the aliasing problem as CS-based model solved by subspace pursuit and present a pruning strategy to further improve the recovery performance and computational cost.

Overall, theoretical complexity analyses and simulation results demonstrate that our sFFT-DT outperforms the state-of-the-art.

%In future work, we want to extend sFFT-DT to a real application, such as GPS \cite{Hassanieh2012acm}.
%In this case, the sparsity constraint is $1$ but $\bm{\hat{x}}$ is interfered with massive noises.
%The most important thing is to identify the location with the maximum magnitude.
%A future direction needing further attention is to combine the mechanism of location estimation in our algorithm and sFFT, as shown in Fig. \ref{fig:CombinedTwoMethod}.
%Fig. \ref{fig:CombinedTwoMethod}(a) shows an extremely sparse signal $\bm{\hat{x}}$ with only one non-zero entry ($K=1$).
%The red region of Fig. \ref{fig:CombinedTwoMethod}(b) is the estimated locations using sFFT (detailed descriptions are ignored due to limited space).
%The red spikes in Fig. \ref{fig:CombinedTwoMethod}(c) are the possible $d$ locations derived from our method.
%By combining both algorithms, the possible location is reduced to the red circle, as shown in Fig. \ref{fig:CombinedTwoMethod}(a).
%Such a new mechanism is more efficient than the one proposed in \cite{Hassanieh2012acm}.
%
%\begin{figure}[h]
%  \centering{\epsfig{figure=SubandDown.eps,width=3.5in}}
%%\hfill
%\caption{The proposed method, sFFT-DT, combined with sFFT \cite{Haitham2012}\cite{Haitham2012_1} can achieve efficient localization in GPS.}
%\label{fig:CombinedTwoMethod}
%\end{figure}

\section{Acknowledgment}
This work was supported by National Science Council under grants NSC 100-2628-E-001-005-MY2 and NSC 102-2221-E-001-022-MY2.
%The authors also would like to thank the anonymous reviewers for their valuable comments that are very constructive and helpful in clarifying the manuscript.

\section{Appendix}\label{Sec: Appendix}
\footnotesize
The analytical solution of solving Step (ii) in syndrome decoding with $a=2$ is
\begin{equation}
\begin{aligned}
    &c_{d}=
    \left|
          \begin{array}{cc}
            m_{0} & m_{1} \\
            m_{1} & m_{2} \\
          \end{array}
        \right|,\\
    \displaystyle &c_{0}= (\frac{1}{c_{d}})\left|
          \begin{array}{cc}
            -m_{2} & m_{1} \\
            -m_{3} & m_{2} \\
          \end{array}
        \right|,\
    \displaystyle c_{1}= (\frac{1}{c_{d}})\left|
          \begin{array}{cc}
            m_{0} & -m_{2} \\
            m_{1} & m_{3} \\
          \end{array}
    \right|,\\
    \displaystyle &z_{0}=\frac{1}{2}[-c_{1}-(c_{1}^{2}-4c_{0})^{\frac{1}{2}}], \
    \displaystyle z_{1}=\frac{1}{2}[-c_{1}+(c_{1}^{2}-4c_{0})^{\frac{1}{2}}],\\
    &p_{d}=z_{1}-z_{0},\\
    \displaystyle &p_{0}= (\frac{1}{p_{d}})\left|
          \begin{array}{cc}
            m_{0} & 1 \\
            m_{1} & z_{1} \\
          \end{array}
    \right|,\
    p_{1}= m_{0}-p_{0}.
\end{aligned}
\label{eq:MPT analytic solutions}
\end{equation}
Similarly, the solution with $a=3$ is
\begin{equation}
\begin{aligned}
    \displaystyle&c_{d}=
    \left|
          \begin{array}{ccc}
            m_{0} & m_{1} & m_{2} \\
            m_{1} & m_{2} & m_{3}\\
            m_{2} & m_{3} & m_{4}\\
          \end{array}
        \right|,\\
    \displaystyle&c_{0}= (\frac{1}{c_{d}})\left|
          \begin{array}{ccc}
            -m_{3} & m_{1} & m_{2} \\
            -m_{4} & m_{2} & m_{3} \\
            -m_{5} & m_{3} & m_{4} \\
          \end{array}
        \right|,\
    c_{1}= (\frac{1}{c_{d}})\left|
          \begin{array}{ccc}
            m_{0} & -m_{3} & m_{1} \\
            m_{1} & -m_{4} & m_{2} \\
            m_{2} & -m_{5} & m_{3} \\
          \end{array}
    \right|,
    c_{2}= (\frac{1}{c_{d}})\left|
          \begin{array}{ccc}
            m_{0} & m_{1} & -m_{3} \\
            m_{1} & m_{2} & -m_{4} \\
            m_{2} & m_{3} & -m_{5} \\
          \end{array}
    \right|,\\
    \displaystyle&z_{0}=-\frac{c_{2}}{3}-A-B,
    z_{1}=-\frac{c_{2}}{3}-W_{1}A-W_{2}B,
    z_{2}=-\frac{c_{2}}{3}-W_{2}A-W_{1}B,\\
    \displaystyle&A=\{ (\frac{c_{0}}{2} - \frac{c_{1}c_{2}}{6}+\frac{c^{3}_{2}}{27})-[ (\frac{c_{0}}{2}-\frac{c_{1}c_{2}}{6}+\frac{c_{2}^{3}}{27}  )^{2} + (\frac{c{1}}{3} - \frac{c_{2}^{2}}{9})^{3} ]^{\frac{1}{2}}  \}^{\frac{1}{3}},\\
    \displaystyle&B=\frac{-(\frac{c_{1}}{3}-\frac{c_{2}^{2}}{9})}{A},\\
    \displaystyle&W_{1}=\frac{-1}{2}+\frac{\sqrt{-3}}{2}, W_{2}=\frac{-1}{2}-\frac{\sqrt{-3}}{2},\\
    \displaystyle&p_{d}=\left|
          \begin{array}{ccc}
            1  & 1  & 1 \\
            z_{0} & z_{1} & z_{2}\\
            z_{0}^{2} & z_{1}^{2} & z_{2}^{2}\\
          \end{array}
    \right|,
    p_{0}= (\frac{1}{p_{d}})\left|
          \begin{array}{ccc}
            m_{0} & 1  & 1 \\
            m_{1} & z_{1} & z_{2}\\
            m_{2} & z_{1}^{2} & z_{2}^{2}\\
          \end{array}
    \right|,
       p_{1}= (\frac{1}{p_{d}})\left|
          \begin{array}{ccc}
            1 & m_{0}  & 1 \\
            z_{0} & m_{1} & z_{2}\\
            z_{0}^{2} & m_{2} & z_{2}^{2}\\
          \end{array}
    \right|,
    p_{2}= 1-p_{0}-p_{1}.
\end{aligned}
\label{eq:MPT analytic solutions with a3}
\end{equation}
Then, the solution with $a=4$ is
\begin{equation*}
\begin{aligned}
    \displaystyle&c_{d}=
    \left|
          \begin{array}{cccc}
            m_{0} & m_{1} & m_{2} & m_{3} \\
            m_{1} & m_{2} & m_{3} & m_{4}\\
            m_{2} & m_{3} & m_{4} & m_{5}\\
            m_{3} & m_{4} & m_{5} & m_{6}\\
          \end{array}
        \right|,\\
    \displaystyle&c_{0}= (\frac{1}{c_{d}})\left|
          \begin{array}{cccc}
            -m_{4} & m_{1} & m_{2} & m_{3}\\
            -m_{5} & m_{2} & m_{3} & m_{4}\\
            -m_{6} & m_{3} & m_{4} & m_{5}\\
            -m_{7} & m_{3} & m_{5} & m_{6}\\
          \end{array}
        \right|,\
    c_{1}= (\frac{1}{c_{d}})\left|
          \begin{array}{cccc}
            m_{0} & -m_{4} & m_{2} & m_{3}\\
            m_{1} & -m_{5} & m_{3} & m_{4}\\
            m_{2} & -m_{6} & m_{4} & m_{5}\\
            m_{3} & -m_{7} & m_{5} & m_{6}\\
          \end{array}
    \right|,\\
    \displaystyle&c_{2}= (\frac{1}{c_{d}})\left|
          \begin{array}{cccc}
            m_{0} & m_{1} & -m_{4} & m_{3} \\
            m_{1} & m_{2} & -m_{5} & m_{4}\\
            m_{2} & m_{3} & -m_{6} & m_{5}\\
            m_{3} & m_{4} & -m_{7} & m_{6}\\
          \end{array}
    \right|,\
       c_{3}= (\frac{1}{c_{d}})\left|
          \begin{array}{cccc}
            m_{0} & m_{1} & m_{2} & -m_{4} \\
            m_{1} & m_{2} & m_{3} & -m_{5}\\
            m_{2} & m_{3} & m_{4} & -m_{6}\\
            m_{3} & m_{4} & m_{5} & -m_{7}\\
          \end{array}
    \right|,
\end{aligned}
\end{equation*}
\begin{equation}
\begin{aligned}
    \displaystyle&z_{0}=\frac{1}{2}\{ -(\frac{c_{3}}{2}+A) - [ (\frac{c_{3}}{2}+A)^{2}-4(Y+B) ]^{\frac{1}{2}}   \},
    z_{1}=\frac{1}{2}\{ -(\frac{c_{3}}{2}+A) + [ (\frac{c_{3}}{2}+A)^{2}-4(Y+B) ]^{\frac{1}{2}}   \},\\
    \displaystyle&z_{2}=\frac{1}{2}\{ -(\frac{c_{3}}{2}-A) - [ (\frac{c_{3}}{2}-A)^{2}-4(Y-B) ]^{\frac{1}{2}}   \},
    z_{3}=\frac{1}{2}\{ -(\frac{c_{3}}{2}-A) + [ (\frac{c_{3}}{2}-A)^{2}-4(Y-B) ]^{\frac{1}{2}}   \},\\
    \displaystyle&A=\frac{1}{2}(c_{3}^2-4c_{2}+8Y)^{\frac{1}{2}},\
    B=\frac{c_{3}Y-c_{1}}{A},\
    Y=\frac{c_{2}}{6}-C-D,\\
    \displaystyle&C=[G+(G^{2}+H^{3})^{\frac{1}{2}}]^{\frac{1}{3}},\
    D=\frac{-H}{C},\
    G=\frac{1}{432}(72c_{0}c_{2}+9c_{1}c_{2}c_{3}-27c_{1}^{2}-27c_{0}c_{3}^2-2c_{2}^{3}),\\
    \displaystyle&H=\frac{1}{36}(3c_{1}c_{3}-12c_{0}-c_{2}^{2}),\\
    \displaystyle&p_{d}=\left|
          \begin{array}{cccc}
            1  & 1  & 1  & 1 \\
            z_{0} & z_{1} & z_{2} & z_{3}\\
            z_{0}^{2} & z_{1}^{2} & z_{2}^{2} & z_{3}^{2}\\
            z_{0}^{3} & z_{1}^{3} & z_{2}^{3} & z_{3}^{3}\\
          \end{array}
    \right|,
    p_{0}= (\frac{1}{p_{d}})\left|
          \begin{array}{cccc}
            1 & 1  & 1  & 1\\
            m_{1} & z_{1} & z_{2} & z_{3}\\
            m_{2} & z_{1}^{2} & z_{2}^{2} & z_{3}^{2}\\
            m_{3} & z_{1}^{3} & z_{2}^{3} & z_{3}^{3}\\
          \end{array}
    \right|,
       p_{1}= (\frac{1}{p_{d}})\left|
          \begin{array}{cccc}
            1 & 1  & 1 & 1 \\
            z_{0} & m_{1} & z_{2} & z_{3}\\
            z_{0}^{2} & m_{2} & z_{2}^{2} & z_{3}^{2}\\
            z_{0}^{3} & m_{3} & z_{2}^{3} & z_{3}^{3}\\
          \end{array}
    \right|,\\
    \displaystyle&p_{2}= (\frac{1}{p_{d}})\left|
          \begin{array}{cccc}
            1 & 1  & 1 & 1 \\
            z_{0} &  z_{1} & m_{1} & z_{3}\\
            z_{0}^{2} &  z_{1}^{2} & m_{2} & z_{3}^{2}\\
            z_{0}^{3} &  z_{1}^{3} & m_{3} & z_{3}^{3}\\
          \end{array}
    \right|,
    p_{3}= 1-p_{0}-p_{1}-p_{2}.
\end{aligned}
\label{eq:MPT analytic solutions with a4}
\end{equation}

\ifCLASSOPTIONcaptionsoff
  \newpage
\fi

\bibliographystyle{IEEEbib}	% (uses file "IEEEbib.bst")
\bibliography{refs}		% expects file "refs.bib"

\end{document}